\documentclass{article}

\usepackage[utf8]{inputenc}
\usepackage{amsthm,mathtools}
\usepackage[margin=1in]{geometry}
\usepackage{amsfonts}
\usepackage{color,soul}
\usepackage{amssymb}
\usepackage{footmisc}
\usepackage[T1]{fontenc}
\usepackage{lmodern}
\usepackage{hyperref}
\usepackage{algpseudocode}
\usepackage{algorithm}
\usepackage{float}
\usepackage{tikz}
\usepackage[framemethod=tikz]{mdframed}
\newtheorem{theorem}{Theorem}[section]
\newtheorem{lemma}[theorem]{Lemma}
\newtheorem{proposition}[theorem]{Proposition}

\newtheorem{corollary}[theorem]{Corollary}
\newtheorem{claim}[theorem]{Claim}
\newtheorem{definition}[theorem]{Definition}
\newcommand{\SC}{O(\log N + (\log r)\cdot\log(1/\epsilon)+(\log r)\cdot \log\log r)} 
\newcommand{\dmax}{d_{\mathrm{max}}}

\newcommand{\dpr}{\text{\textcircled{p}}_{\mathcal{H}}}
\newcommand{\R}{\mathbb{R}}

\newcommand{\Space}{\mathrm{Space}}
\newcommand{\class}[1]{\mathbf{#1}}

\newcommand{\RL}{\class{RL}}

\renewcommand{\L}{\class{L}}

\newcommand{\poly}{\mathrm{poly}}
\newcommand{\polylog}{\mathrm{polylog}}

\def\textprob#1{\textmd{\textsc{#1}}}
\newcommand{\USTConn}{\textprob{Undirected S-T Connectivity}}

\newcommand*\samethanks[1][\value{footnote}]{\footnotemark[#1]}

\newcounter{myalgctr}
\newenvironment{myalg}{%      define a custom environment
   \bigskip\noindent%         create a vertical offset to previous material
   \refstepcounter{myalgctr}% increment the environment's counter
   \newline%
   }{\par\bigskip}  %          create a vertical offset to following material
\numberwithin{myalgctr}{section}
\newtheoremstyle{named}{}{}{\itshape}{}{\bfseries}{.}{.5em}{\thmnote{#3}}
\theoremstyle{named}
\newtheorem{namedtheorem}{Theorem}[section]

\title{Deterministic Approximation of Random Walks in Small Space}
\date{\today}
\author{Jack Murtagh\thanks{Supported by NSF grant CCF-1763299.} \\School of Engineering \& Applied Sciences\\
Harvard University\\ 
Cambridge, MA USA\\
\texttt{jmurtagh@g.harvard.edu}\\
\url{http://scholar.harvard.edu/jmurtagh} \and Omer Reingold\thanks{Supported by NSF grant CCF-1763311.} \\ Computer Science Department\\ Stanford University\\ Stanford, CA USA \\ \texttt{reingold@stanford.edu}  \and Aaron Sidford \\ Management Science \& Engineering \\ Stanford University \\ Stanford, CA USA \\ \texttt{sidford@stanford.edu} \\ \url{http://www.aaronsidford.com} \and 
Salil Vadhan\samethanks[1]\\School of Engineering \& Applied Sciences\\
Harvard University\\ 
Cambridge, MA USA\\
\texttt{salil\_vadhan@harvard.edu}\\
\url{http://salil.seas.harvard.edu/}}
\begin{document}

\maketitle
\begin{abstract}
We give a deterministic, nearly logarithmic-space algorithm that given an undirected graph $G$, a positive integer $r$, and a set $S$ of vertices, approximates the conductance of $S$ in the $r$-step random walk on $G$ to within a factor of $1+\epsilon$, where $\epsilon>0$ is an arbitrarily small constant. More generally, our algorithm computes an $\epsilon$-spectral approximation to the normalized Laplacian of the $r$-step walk.

Our algorithm combines the derandomized square graph operation \cite{RozenmanVa05}, which we recently used for solving Laplacian systems in nearly logarithmic space \cite{derandbeyond2017}, with ideas from \cite{cheng2015}, which gave an algorithm that is time-efficient (while ours is space-efficient) and randomized (while ours is deterministic) for the case of even $r$ (while ours works for all $r$). Along the way, we provide some new results that generalize technical machinery and yield improvements over previous work. First, we obtain a nearly linear-time randomized algorithm for computing a spectral approximation to the  normalized Laplacian for odd $r$. Second, we define and analyze a generalization of the derandomized square for irregular graphs and for sparsifying the product of two distinct graphs. As part of this generalization, we also give a strongly explicit construction of expander graphs of every size.
\end{abstract}

\newpage
\section{Introduction}
\label{sect:intro}
Random walks provide the most dramatic example of the power of randomized algorithms for solving computational problems in the space-bounded setting, as they only require logarithmic space (to store the current state or vertex). In particular, since undirected graphs have polynomial cover time, random walks give a randomized logspace ($\RL$) algorithm for $\USTConn$ \cite{AleliunasKaLiLoRa79}. Reingold \cite{Reingold08} showed that this algorithm can be derandomized, and hence that $\USTConn$ is in deterministic logspace ($\L$). However, Reingold's algorithm does not match the full power of random walks on undirected graphs; in particular it does not allow us to approximate properties of the random walk at lengths below the mixing time. 

In this work, we provide a nearly logarithmic-space algorithm for approximating properties of arbitrary-length random walks on an undirected graph, in particular the \emph{conductance} of any set of vertices:

\begin{definition}
Let $G=(V,E)$ be an undirected graph, $r$ a positive integer, and $S\subseteq V$ a set of vertices.  The {\em conductance} of $S$ under the $r$-step random walk on $G$ is defined as 
\[
\Phi_r(S)=\Pr[V_r\not \in S | V_0\in S],
\]
where $V_0,V_1,\ldots,V_r$ is a random walk on $G$ started at the stationary distribution  $\Pr[V_0=v]=\deg(v)/2|E|$.
\end{definition}

\begin{theorem}
\label{thm:conductance}
There is a deterministic algorithm that given an undirected multigraph $G$ on $n$ vertices, a positive integer $r$, a set of vertices $S$, and $\epsilon>0$, computes a number $\tilde{\Phi}$ such that
\[
(1-\epsilon)\cdot \Phi_r(S) \leq \tilde{\Phi} \leq (1+\epsilon)\cdot \Phi_r(S)
\]
and runs in space $\SC$, where $N$ is the bit length of the input graph $G$. 
\end{theorem}

Previously, approximating conductance could be done in $O(\log^{3/2}(N/\epsilon)+\log\log r)$ space, which follows from Saks' and Zhou's proof that $\RL$ is in $\L^{3/2}$ \cite{SaksZh99}. 

Two interesting parameter regimes where we improve the Saks-Zhou bound are when $r=1/\epsilon=2^{O(\sqrt{\log N})}$, in which case our algorithm runs in space $O(\log N)$, or when $\epsilon=1/\mathrm{polylog}(N)$ and $r\leq \mathrm{poly}(N)$, in which case our algorithm runs in space $\tilde{O}(\log N)$. When $r$ exceeds the $\mathrm{poly}(N)\cdot\log(1/\epsilon)$ time for random walks on undirected graphs to mix to within distance $\epsilon$ of the stationary distribution, the conductance can be approximated in space $O(\log (N/\epsilon)+\log\log r)$ by using Reingold's algorithm to find the connected components of $G$, and the bipartitions of the components that are bipartite and calculating the stationary probability of $S$ restricted to each of these pieces, which is proportional to the sum of degrees of vertices in $S$.  

We prove Theorem~\ref{thm:conductance} by providing a stronger result that with the same amount of space it is possible to compute an \emph{$\epsilon$-spectral approximation} to the \emph{normalized Laplacian} of the $r$-step random walk on $G$. 
\begin{definition}
Let $G$ be an undirected graph with adjacency matrix $A$, diagonal degree matrix $D$, and transition matrix $T=AD^{-1}$. The \emph{transition matrix for the $r$-step random walk} on $G$ is $T^r$. The \emph{normalized Laplacian} of the $r$-step random walk is the symmetric matrix $I-M^r$ for $M=D^{-1/2}AD^{-1/2}$. 
\end{definition}
Note that the normalized Laplacian can also be expressed as
$I-M^r = D^{-1/2}(I-T^r)D^{1/2}$, so it does indeed capture the behavior of $r$-step random walks on $G$.\footnote{When $G$ is irregular, the matrix $I-T^r$ is not necessarily symmetric. It is a \emph{directed Laplacian} as defined in \cite{CKPPSV16, CKPPRSV16}. See Definition \ref{def:dirlap}.} 

\begin{theorem}[Main result]
\label{thm:main}
There is a deterministic algorithm that given an undirected multigraph $G$ on $n$ vertices with normalized Laplacian $I-M$, a nonnegative integer $r$, and $\epsilon>0$, constructs an undirected multigraph $\tilde{G}$ whose normalized Laplacian $\tilde{L}$ is an $\epsilon$-spectral approximation of $L=I-M^r$. That is, for all vectors $v\in\mathbb{R}^n$
\[
(1-\epsilon)\cdot v^{T}Lv \leq v^{T}\tilde{L}v \leq (1+\epsilon)\cdot v^{T}Lv.
\]
The algorithm runs in space $\SC$, where $N$ is the bit length of the input graph $G$. 
\end{theorem}
Theorem \ref{thm:conductance} follows from Theorem \ref{thm:main} by taking $v$ to be $D^{1/2}e_S$ where $e_S$ is the characteristic vector of the set $S$ and normalizing appropriately (See Section \ref{sect:cors}). 

Our main technique for proving Theorem \ref{thm:main} is the \emph{derandomized product}, a new generalization of the \emph{derandomized square}, which was introduced by Rozenman and Vadhan \cite{RozenmanVa05} to give an alternative proof that \\$\USTConn$ is in $\L$. Our main result follows from carefully applying the derandomized product and analyzing its properties with inequalities from the theory of spectral approximation. Specifically, our analysis is inspired by the work of Cheng, Cheng, Liu, Peng, and Teng \cite{cheng2015}, who studied the approximation of random walks by randomized algorithms running in nearly linear time. We emphasize that the work of \cite{cheng2015} gives a randomized algorithm with high space complexity (but low time complexity) for approximating properties of even length walks while we give a deterministic, space-efficient algorithm for approximating properties of walks of every length. Interestingly, while the graphs in our results are all undirected, some of our analyses use techniques for spectral approximation of \emph{directed} graphs introduced by Cohen, Kelner, Peebles, Peng, Rao, Sidford, and Vladu \cite{CKPPSV16,CKPPRSV16}. 

The derandomized square can be viewed as applying the pseudorandom generator of Impagliazzo, Nisan, and Wigderson \cite{ImpagliazzoNiWi94} to random walks on labelled graphs. It is somewhat surprising that repeated derandomized squaring does not blow up the error by a factor proportional to the length of the walk being derandomized. For arbitrary branching programs, the INW generator does incur error that is linear in the length of the program. Some special cases such as  regular \cite{BravermanRaRaYe10, BrodyVe10, De11} and permutation \cite{De11, Steinke12} branching programs of constant width have been shown to have a milder error growth as a function of the walk length. Our work adds to this list by showing that properties of random walks of length $k$ on undirected graphs can be estimated in terms of spectral approximation without error accumulating linearly in $k$.

In our previous work \cite{derandbeyond2017}, we showed that the Laplacian of the derandomized square of a \emph{regular} graph spectrally approximates the Laplacian of the true square, $I-M^2$, and this was used in a recursion from \cite{PS13} to give a nearly logarithmic-space algorithm for approximately solving Laplacian systems $Lx=b$. A natural idea to approximate the Laplacian of higher powers, $I-M^r$, is to repeatedly derandomized square. This raises three challenges, and we achieve our result by showing how to overcome each: 
\begin{enumerate}
\item It is not guaranteed from \cite{derandbeyond2017} that repeated derandomized squaring preserves spectral approximation. For this, we use ideas from \cite{cheng2015} to argue that it does.
\item When $r$ is not a power of 2, the standard approach would be to write $r=b_0+2\cdot b_1+\ldots+2^z\cdot b_z$ where $b_i$ is the $i$th bit of $r$ and multiply approximations to $M^{2^{i}}$ for all $i$ such that $b_i\neq 0$. The problem is that multiplying spectral approximations of matrices does not necessarily yield a spectral approximation of their product. Our solution is to generalize the derandomized square to produce sparse approximations to the product of \emph{distinct} graphs. In particular, given $I-M$ and an approximation $I-\tilde{M}$ to $I-M^k$, our derandomized product allows us to combine $M$ and $\tilde{M}$ to approximate $I-M^{k+1}$. Although our generalized graph product is defined for undirected graphs, its analysis uses machinery for spectral approximation of directed graphs, introduced in \cite{CKPPRSV16}.
\item We cannot assume that our graph is regular without loss of generality. In contrast, \cite{Reingold08,RozenmanVa05,derandbeyond2017} could do so, since adding self-loops does not affect connectivity or solutions to Laplacian systems of $G$, however, it does affect random walks. Our solution is to define and analyze the derandomized product for irregular graphs.
\end{enumerate}

A key element in the derandomized product is a strongly explicit (i.e. neighbor relations can be computed in time $\polylog (N)$ and space $O(\log N)$) construction of expander graphs whose sizes equal the degrees of the vertices in the graphs being multiplied. Many commonly used strongly explicit expander families only contain graph sizes that are certain subsets of $\mathbb{N}$ such as powers of 2 (Cayley graphs based on \cite{NaorNa93} and \cite{AlonGoHaPe92}), perfect squares \cite{Margulis73, GabberGa81}, and other size distributions \cite{ReingoldVaWi01} or are only mildly explicit in the sense of having running time or parallel work that is $\poly(N)$ \cite{LeePS15}.  It was folklore that one can convert any sufficiently dense family of expanders (such as the aforementioned constructions) into a family of expanders of every size, and there were offhand remarks in the literature to this effect (cf. \cite{alon2008elementary}, \cite{goldreich08}), albeit without an analysis of expansion or explicitness. In Section \ref{sect:expanders}, we give a self-contained description and proof of such a reduction, with an analysis in terms of both strong explicitness and spectral expansion. Subsequent to our work, Goldreich \cite{gnote} has provided an analysis of the construction previously described in \cite{goldreich08} in terms of vertex expansion and strong explicitness, and Alon \cite{alon_inprep} has given a stronger reduction that provides nearly Ramanujan expanders (i.e. ones with nearly optimal spectral expansion as a function of the degree) of every size $n$.

Many of our techniques are inspired by Cheng, Cheng, Liu, Peng, and Teng \cite{cheng2015}, who gave two algorithms for approximating random walks. One is a nearly linear time randomized algorithm for approximating random walks of \emph{even} length and another works for all walk lengths $r$ but has a running time that is quadratic in $r$, and so only yields a nearly linear time algorithm for $r$ that is polylogarithmic in the size of the graph. In addition, \cite{JKPS17} studied the problem of computing sparse spectral approximations of random walks but the running time in their work also has a quadratic dependence on $r$. We extend these results by giving a nearly linear time randomized algorithm for computing a spectral approximation to $I-M^r$ for \emph{all} $r$. This is discussed in Section \ref{sect:cors}.

\section{Preliminaries}
\label{sect:prelims}
\subsection{Spectral Graph Theory}
Given an undirected multigraph $G$ the \emph{Laplacian} of $G$ is the symmetric matrix $D-A$, where $D$ is the diagonal matrix of vertex degrees and $A$ is the adjacency matrix of $G$. The \emph{transition matrix} of the random walk on $G$ is $T=AD^{-1}$. $T_{ij}$ is the probability that a uniformly random edge from vertex $j$ leads to vertex $i$ (i.e. the number of edges between $j$ and $i$ divided by the degree of $j$). The \emph{normalized Laplacian} of $G$ is the symmetric matrix $I-M=D^{-1/2}(D-A)D^{-1/2}$. Note that when $G$ is regular, the matrix $M=D^{-1/2}AD^{-1/2}=AD^{-1}=T$. The \emph{transition matrix of the $r$-step random walk} on $G$ is $T^r$. For all probability distributions $\pi$, $T^r\pi$ is the distribution over vertices that results from picking a random vertex according to $\pi$ and then running a random walk on $G$ for $r$ steps. The transition matrix of the $r$-step random walk on $G$ is related to the normalized Laplacian in the following way:
\[
I-M^r = D^{-1/2}(I-T^r)D^{1/2}.
\]
For undirected multigraphs, the matrix $M=D^{-1/2}AD^{-1/2}$ has real eigenvalues between $-1$ and $1$ and so $I-M^r$ has eigenvalues in $[0,2]$ and thus is positive semidefinite (PSD). The \emph{spectral norm} of a real matrix $M$, denoted $\|M\|$, is the largest singular value of $M$. That is, the square root of the largest eigenvalue of $M^TM$. When $M$ is symmetric, $\|M\|$ equals the largest eigenvalue of $M$ in absolute value. For an undirected graph $G$ with adjacency matrix $A$, we write $k\cdot G$ to denote the graph with adjacency matrix $k\cdot A$, i.e. the multigraph $G$ with all edges duplicated to have multiplicity $k$. 

Given a symmetric matrix $L$, its \emph{Moore-Penrose Pseudoinverse}, denoted $L^{\dagger}$, is the unique matrix with the same eigenvectors as $L$ such that for each eigenvalue $\lambda$ of $L$, the corresponding eigenvalue of $L^{\dagger}$ is $1/\lambda$ if $\lambda\neq 0$ and 0 otherwise. When $L$ is a Laplacian, we write $L^{\dagger/2}$ to denote the unique symmetric PSD matrix square root of the pseudoinverse of $L$. 

To measure the approximation between graphs we use spectral approximation\footnote{In \cite{derandbeyond2017}, we use an alternative definition of spectral approximation where $\tilde{L}\approx_{\epsilon}L$ if for all $v\in\mathbb{R}^n$, $e^{-\epsilon}\cdot v^{T}Lv\leq v^T\tilde{L}v\leq e^{\epsilon}\cdot v^{T}Lv$. We find Definition \ref{def:spectralapprox} more convenient for this paper.}\cite{ST04}:
\begin{definition}
\label{def:spectralapprox}
Let $L, \tilde{L}\in\mathbb{R}^{n\times n}$ be symmetric PSD matrices. We say that \emph{$\tilde{L}$ is an \emph{$\epsilon$-approximation} of $L$} (written $\tilde{L}\approx_{\epsilon}L$) if for all vectors $v\in\mathbb{R}^n$
\[
(1-\epsilon)\cdot v^{T}Lv \leq v^{T}\tilde{L}v \leq (1+\epsilon)\cdot v^{T}Lv.
\]
\end{definition}
Note that Definition \ref{def:spectralapprox} is not symmetric in $L$ and $\tilde{L}$. Spectral approximation can also be written in terms of the Loewner partial ordering of PSD matrices:
\[
(1-\epsilon)\cdot L \preceq \tilde{L} \preceq (1+\epsilon)\cdot L
\]
where for two matrices $A,B$, we write $A\preceq B$ if $B-A$ is PSD. Spectral approximation has a number of useful properties listed in the following proposition. 

\begin{proposition}
\label{prop:psdfacts}
If $W,X, Y,Z\in \R^{n\times n}$ are PSD symmetric matrices then: 
\begin{enumerate}
\item If $X\approx_{\epsilon}Y$ for $\epsilon<1$ then $Y\approx_{\epsilon/(1-\epsilon)}X$
\item If $X\approx_{\epsilon_1}Y$ and $Y\approx_{\epsilon_2}Z$ then $X\approx_{\epsilon_1 + \epsilon_2+\epsilon_1\cdot\epsilon_2}Z$ ,
\item If $X\approx_{\epsilon}Y$ and $V$ is any $n\times n$ matrix then $V^{T}XV\approx_{\epsilon}V^{T}YV$,
\item If $X\approx_{\epsilon}Y$ then $X+Z\approx_{\epsilon} Y+Z$,
\item If $W\approx_{\epsilon_1} X$ and $Y\approx_{\epsilon_2} Z$ then $W+Y\approx_{\max\{\epsilon_1,\epsilon_2\}} X+Z$, and
\item If $X\approx_{\epsilon}Y$ then $c\cdot X\approx_{\epsilon}c\cdot Y$ for all nonnegative scalars $c$
\end{enumerate}
\end{proposition}
For regular undirected graphs, we use the measure introduced by \cite{Mihail89} for the rate at which a random walk converges to the uniform distribution.
\begin{definition}[\cite{Mihail89}] 
\label{def:expansion}
Let $G$ be a regular undirected graph with transition matrix $T$. Define
\[
\lambda(G)\coloneqq\max_{\substack{v\perp\vec{1}\\ v\neq 0}}\frac{\|Tv\|}{\|v\|}=\mathrm{2nd~largest ~absolute~value ~of ~the~ eigenvalues~ of ~}T\in[0,1].
\]
$1-\lambda(G)$ is called the \emph{spectral gap} of $G$.
\end{definition}
$\lambda(G)$ is known to be a measure of how well-connected a graph is. The smaller $\lambda(G)$, the faster a random walk on $G$ converges to the uniform distribution. Graphs $G$ with $\lambda(G)$ bounded away from 1 are called \emph{expanders}. Expanders can equivalently be characterized as graphs that spectrally approximate the complete graph. This is formalized in the next lemma.

\begin{lemma}
\label{lem:expanderapprox}
Let $H$ be a $c$-regular undirected multigraph on $n$ vertices with transition matrix $T$ and let $J\in\mathbb{R}^{n\times n}$ be a matrix with $1/n$ in every entry (i.e. $J$ is the transition matrix of the complete graph with a self loop on every vertex). Then $\lambda(H)\leq \lambda$ if and only if $I-T\approx_{\lambda} I-J$.
\end{lemma}
A proof of Lemma \ref{lem:expanderapprox} can be found in Appendix \ref{app:expanderapprox}. In \cite{CKPPRSV16} Cohen, Kelner, Peebles, Peng, Rao, Sidford, and Vladu introduced a definition of spectral approximation for \emph{asymmetric matrices}. Although the results in our paper only concern undirected graphs, some of our proofs use machinery from the theory of directed spectral approximation.  

\begin{definition}[Directed Laplacian \cite{CKPPSV16,CKPPRSV16}]
\label{def:dirlap}
A matrix $L\in\mathbb{R}^{n\times n}$ is called a \emph{directed Laplacian} if $L_{ij}\leq 0$ for all $i\neq j$ and $L_{ii}=-\sum_{j\neq i}L_{ji}$ for all $i\in[n]$. The associated directed graph has $n$ vertices and an edge $(i,j)$ of weight $-L_{ji}$ for all $i\neq j\in[n]$ with $L_{ji}\neq 0$.
\end{definition}

\begin{definition}[Asymmetric Matrix Approximation \cite{CKPPRSV16}]
\label{def:dirapprox}
Let $\tilde{L}$ and $L$ be (possibly asymmetric) matrices such that $U=(L+L^T)/2$ is PSD. We say that $\tilde{L}$ is a \emph{directed} $\epsilon$-\emph{approximation} of $L$ if:
\begin{enumerate}
\item $\mathrm{ker}(U)\subseteq \mathrm{ker}(\tilde{L}-L)\cap\mathrm{ker}((\tilde{L}-L)^T)$, and
\item $\left\|U^{\dagger/2}(\tilde{L}-L)U^{\dagger/2}\right\|_2 \leq \epsilon$
\end{enumerate}
\end{definition}
Below we state some useful lemmas about directed spectral approximation. The first gives an equivalent formulation of Definition \ref{def:dirapprox}.

\begin{lemma}[\cite{CKPPRSV16} Lemma 3.5]
\label{lem:dirapprox}
Let $L\in\mathbb{R}^{n\times n}$ be a (possibly asymmetric) matrix and let $U=(L+L^T)/2$. A matrix $\tilde{L}$ is a directed $\epsilon$-approximation of $L$ if and only if for all vectors $x,y\in\mathbb{R}^n$
\[
x^T(\tilde{L}-L)y\leq \frac{\epsilon}{2}\cdot(x^TUx + y^TUy).
\]
\end{lemma}

\begin{lemma}[\cite{CKPPRSV16} Lemma 3.6]
\label{lem:dir_implies_undir}
Suppose $\tilde{L}$ is a directed $\epsilon$-approximation of $L$ and let $U=(L+L^T)/2$ and $\tilde{U}=(\tilde{L}+\tilde{L}^T)/2$. Then $\tilde{U}\approx_{\epsilon}U$.
\end{lemma}

Lemma \ref{lem:dir_implies_undir} says that directed spectral approximation implies the usual notion from Definition \ref{def:spectralapprox} for ``symmetrized'' versions of the matrices $L$ and $\tilde{L}$. In fact, when the matrices $L$ and $\tilde{L}$ are both symmetric, the two definitions are equivalent:
\begin{lemma}
\label{lem:symmequiv}
Let $\tilde{L}$ and $L$ be symmetric PSD matrices. Then $\tilde{L}$ is a directed $\epsilon$-approximation of $L$ if and only if $\tilde{L}\approx_{\epsilon}L$.
\end{lemma}
A proof of Lemma \ref{lem:symmequiv} can be found in Appendix \ref{app:symmequiv}.

\subsection{Space Bounded Computation}
\label{sect:spacemodel}
We use a standard model of space-bounded computation where the machine $\mathcal{M}$ has a read-only input tape, a constant number of  read/write work tapes, and a write-only output tape. If throughout every computation on inputs of length at most $n$, $\mathcal{M}$ uses at most $s(n)$ total tape cells on all the work tapes, we say $\mathcal{M}$ runs in space $s=s(n)$. Note that $\mathcal{M}$ may write more than $s$ cells (in fact as many as $2^{O(s)}$) but the output tape is write-only. The following proposition describes the behavior of space complexity when space bounded algorithms are composed. 
\begin{proposition}
\label{prop:composition}
Let $f_1$, $f_2$ be functions that can be computed in space $s_1(n),s_2(n)\geq \log n$, respectively, and $f_1$ has output of length at most $\ell_1(n)$ on inputs of length $n$. Then $f_2\circ f_1$ can be computed in space 
\[
O(s_2(\ell_1(n)) + s_1(n)).
\]
\end{proposition}

\subsection{Rotation Maps}
In the space-bounded setting, it is convenient to use local descriptions of graphs. Such descriptions allow us to navigate large graphs without loading them entirely into memory. For this we use $\emph{rotation maps}$, functions that describe graphs through their neighbor relations. Rotation maps are defined for graphs with labeled edges as described in the following definition.

\begin{definition}[\cite{ReingoldVaWi01}]
\label{def:labeling}
A {\em two-way labeling} of an undirected multigraph $G=(V,E)$ with vertex degrees $(d_v)_{v\in V}$, is a labeling of the edges in $G$ such that 
\begin{enumerate}
\item Every edge $(u,v)\in E$ has two labels: one in $[d_u]$ as an edge incident to $u$ and one in $[d_v]$ as an edge incident to $v$,
\item For every vertex $v\in V$, the labels of the $d_v$ edges incident to $v$ are distinct.
\end{enumerate}
\end{definition}
In \cite{RozenmanVa05}, two-way labelings are referred to as {\em undirected} two-way labelings. Note that every graph has a two-way labeling where each vertex ``names'' its neighbors uniquely in some canonical way based on the order they're represented in the input. We will describe multigraphs with two-way labelings using rotation maps:

\begin{definition}[\cite{ReingoldVaWi01}]
\label{def:rotmap}
Let $G$ be an undirected multigraph on $n$ vertices with a two-way labeling. The {\em rotation map} Rot$_{G}$ is defined as follows: Rot$_{G}(v,i)=(w,j)$ if the $i$th edge to vertex $v$ leads to vertex $w$ and this edge is the $j$th edge incident to $w$. 
\end{definition}

We will use expanders that have efficiently computable rotation maps. We call such graphs \emph{strongly explicit}. The usual definition of strong explicitness only refers to time complexity, but we will use it for both time and space. 
\begin{definition}
A family of two-way labeled graphs $\mathcal{G}=\{G_{n,c}\}_{(n,c)}$, where $G_{n,c}$ is a $c$-regular graph on $n$ vertices, is called \emph{strongly explicit} if given $n,c$, a vertex $v\in[n]$ and an edge label $a\in[c]$, Rot$_{G_{n,c}}(v,a)$ can be computed in time $\mathrm{poly}(\log (nc))$ and space $O(\log nc)$.
\end{definition}

\section{The Derandomized Product and Expanders of All Sizes}
\label{sect:expanders}

In this section we introduce our derandomized graph product. The derandomized product generalizes the \emph{derandomized square} graph operation that was introduced by Rozenman and Vadhan \cite{RozenmanVa05} to give an alternative proof that $\USTConn$ is in $\L$. Unlike the derandomized square, the derandomized product is defined for \emph{irregular} graphs and produces a sparse approximation to the product of any two (potentially different) graphs with the same vertex degrees.

Here, by the `product' of two graphs $G_0,G_1$, we mean the reversible Markov chain with transitions defined as follows: from a vertex $v$, with probability $1/2$ take a random step on $G_0$ followed by a random step on $G_1$ and with probability $1/2$ take a random step on $G_1$ followed by a random step on $G_0$. 

When $G_0=G_1=G$, this is the same as taking a 2-step random walk on $G$. Note, however, that when $G$ is irregular, a 2-step random walk is \emph{not} equivalent to doing a 1-step random walk on the graph $G^2$, whose edges correspond to paths of length 2 in $G$. Indeed, even the stationary distribution of the random walk on $G^2$ may be different than on $G$.\footnote{For example, let $G$ be the graph on two vertices with one edge $(u,v)$ connecting them and a single self loop on $u$. Then $[2/3,1/3]$ is the stationary distribution of $G$ and $[3/5,2/5]$ is the stationary distribution of $G^2$.} Nevertheless, our goal in the derandomized product is to produce a relatively sparse graph whose 1-step random walk approximates the 2-step random walk on $G$. 

The intuition behind the derandomized product is as follows: rather than build a graph with every such two-step walk, we use expander graphs to pick a pseudorandom subset of the walks.  Specifically, we first pick $b\in\{0,1\}$ at random. Then, as before we take a truly random step from $v$ to $u$ in $G_b$. But for the second step, we don't use an arbitrary edge leaving $u$ in $G_{\bar{b}}$, but rather correlate it to the edge on which we arrived at $u$ using a $c$-regular expander on deg$(u)$ vertices, where we assume that the vertex degrees in $G_0$ and $G_1$ are the same. When $c<\mathrm{deg}(u)$, the vertex degrees of the resulting two-step graph will be sparser than without derandomization. However using the pseudorandom properties of expander graphs, we can argue that the derandomized product is a good approximation of the true product. 

\begin{definition}[Derandomized Product]
\label{def:derandproduct}
Let $G_0, G_1$ be undirected multigraphs on $n$ vertices with two-way labelings and identical vertex degrees $d_1,d_2,\ldots,d_n$. Let $\mathcal{H}=\{H_i\}$ be a family of two-way labeled, $c$-regular expanders of sizes including $d_1,\ldots,d_n$. The {\em derandomized product with respect to $\mathcal{H}$}, denoted $G_0\dpr G_1$, is an undirected multigraph on $n$ vertices with vertex degrees $2\cdot c\cdot d_1,\ldots,2\cdot c\cdot d_n$ and rotation map Rot$_{G_0\dpr G_1}$ defined as follows: For $v_{0}\in[n], j_0\in[d_{v_0}]$, $a_0\in [c]$, and $b\in\{0,1\}$ we compute Rot$_{G_0\dpr G_1}(v_0,(j_0,a_0,b))$ as
\begin{enumerate}
\item Let $(v_1,j_1)=$Rot$_{G_b}(v_0,j_0)$
\item Let $(j_2,a_1)=$Rot$_{H_{d_{v_{1}}}}(j_1,a_0)$
\item Let $(v_2,j_3)=$Rot$_{G_{\bar{b}}}(v_1,j_2)$
\item Output $(v_2,(j_3,a_1,\bar{b}))$
\end{enumerate}
where $\bar{b}$ denotes the bit-negation of $b$. 
\end{definition}
Note that when $G_0=G_1$ the derandomized product generalizes the derandomized square \cite{RozenmanVa05} to irregular graphs, albeit with each edge duplicated twice. Note that Definition \ref{def:derandproduct} requires that each vertex $i$ has the same degree $d_i$ in $G_0$ and $G_1$, ensuring that the random walks on $G_0,G_1$, and $G_0\dpr G_1$ all have the same stationary distribution. This can be generalized to the case that there is an integer $k$ such that for each vertex $v$ with degree $d_v$ in $G_1$, $v$ has degree $k\cdot d_v$ in $G_0$. For this, we can duplicate each edge in $G_1$ $k$ times to match the degrees of $G_0$ and then apply the derandomized product to the result. In such cases we abuse notation and write $G_0\dpr G_1$ to mean $G_0\dpr k\cdot G_1$.  

In \cite{derandbeyond2017} we showed that the derandomized square produces a spectral approximation to the true square. We now show that the derandomized product also spectrally approximates a natural graph product. 

\begin{theorem}
\label{thm:dprapprox}
Let $G_0,G_1$ be undirected multigraphs on $n$ vertices with two-way labelings, and normalized Laplacians $I-M_0$ and $I-M_1$. Let $G_0$ have vertex degrees $d_1,\ldots,d_n$ and $G_1$ have vertex degrees $d'_1,\ldots,d'_n$ where for all $i\in[n]$, $d_i=k\cdot d'_i$ for a positive integer $k$. Let $\mathcal{H}=\{H_i\}$ be a family of two-way labeled, $c$-regular expanders with $\lambda(H_i)\leq\lambda$ for all $H_i\in\mathcal{H}$, of sizes including $d_1,\ldots,d_n$. Let $I-\tilde{M}$ be the normalized Laplacian of $\tilde{G}=G_0\dpr G_1$. Then
\[
I-\tilde{M}\approx_{\lambda} I-\frac{1}{2}\cdot(M_0M_1+M_1M_0).
\]
\end{theorem}

\begin{proof}[Proof of Theorem \ref{thm:dprapprox}]
Note that $k\cdot G_1$ has the same transition matrix and normalized Laplacian as $G_1$. So we can replace $G_1$ with $k\cdot G_1$ and assume $k=1$ without loss of generality. 

Since $G_0$ and $G_1$ have the same vertex degrees, we can we write
\begin{equation}
\label{eq:1}
I-\frac{1}{2}\cdot(M_0M_1+M_1M_0) = I-D^{-1/2}\frac{1}{2}\cdot(T_0T_1+T_1T_0)D^{1/2}
\end{equation}
where $T_0$ and $T_1$ are the transition matrices of $G_0$ and $G_1$, respectively.

Following the proofs in \cite{RozenmanVa05} and \cite{derandbeyond2017}, we can write the transition matrix for the random walk on $\tilde{G}$ as $\tilde{T}=\frac{1}{2}\cdot(P R_0 \tilde{B} R_1 Q + P R_1 \tilde{B} R_0 Q)$, where each matrix corresponds to a step in the definition of the derandomized product. The two terms correspond to $b=0$ and $b=1$ in the derandomized product and, setting $\bar{d}=\sum_{i\in[n]}d_i$,

\begin{itemize}
\item $Q$ is a $\bar{d}\times n$ matrix that ``lifts'' a probability distribution over $[n]$ to one over $[\bar{d}]$ where the mass on each coordinate $i\in[n]$ is divided uniformly over the corresponding degree $d_i$. That is, $Q_{(u,i),v}=1/d_{i}$ if $u=v$ and 0 otherwise where the rows of $Q$ are ordered $(1,1),(1,2),\ldots,(1,d_1),(2,1),\ldots,(2,d_2),\\ \ldots(n,1),\ldots,(n,d_n)$.
\item $R_0$ and $R_1$ are the $\bar{d}\times \bar{d}$ symmetric permutation matrices corresponding to the rotation maps of $G_0$ and $G_1$, respectively. That is, entry $(u,i),(v,j)$ in $R_a$ is $1$ if Rot$_{G_{a}}(u,i)=(v,j)$ and 0 otherwise for $a\in\{0,1\}$.
\item $\tilde{B}$ is a $\bar{d}\times \bar{d}$ symmetric block-diagonal matrix with $n$ blocks where block $i$ is the transition matrix for the random walk on $H_{d_{i}}\in\mathcal{H}$, the expander in our family with $d_i$ vertices.
\item $P=DQ^{T}$ is the $n\times \bar{d}$ matrix that maps any $\bar{d}$-vector to an $n$-vector by summing all the entries corresponding to edges incident to the same vertex in $G_0$ and $G_1$. This corresponds to projecting a distribution on $[\bar{d}]$ back down to a distribution over $[n]$. $P_{v,(u,i)}=1$ if $u=v$ and 0 otherwise where the columns of $P$ are ordered $(1,1),(1,2),\ldots,(1,d_1),(2,1),\ldots,(2,d_2),\ldots(n,1),\ldots,(n,d_n)$.
\end{itemize}
Likewise, we can write
\begin{equation}
\label{eq:2}
(T_0T_1+T_1T_0)=(P R_0 \tilde{J} R_1 Q + P R_1 \tilde{J} R_0 Q)
\end{equation}
where $\tilde{J}$ is a $\bar{d}\times \bar{d}$ symmetric block-diagonal matrix with $n$ blocks where block $i$ is $J_i$, the transition matrix for the complete graph on $d_i$ vertices with a self loop on every vertex. That is, every entry of $J_i$ is $1/d_i$.

We will show that
\[
I_{\bar{d}}-\frac{1}{2}\cdot(R_0\bar{B}R_1+R_1\bar{B}R_0)\approx_{\lambda} I_{\bar{d}}-\frac{1}{2}\cdot(R_0\bar{J}R_1+R_1\bar{J}R_0).
\]
From this the theorem follows by multiplying by $D^{-1/2}P$ on the left and $(D^{-1/2}P)^T=QD^{1/2}$ on the right and applying Proposition \ref{prop:psdfacts} Part 3. Since $D^{-1/2}PQD^{1/2}=I_n$, the left-hand side becomes 
\begin{align*}
I_n-D^{-1/2}\tilde{T}D^{1/2}&=I_n-\tilde{D}^{-1/2}\tilde{T}\tilde{D}^{1/2}\\
&=I_n-\tilde{M}
\end{align*}
where $\tilde{D}=2\cdot c\cdot D$ is the diagonal matrix of vertex degrees of $\tilde{G}$. By Equations (\ref{eq:1}) and (\ref{eq:2}), the right-hand side becomes $I_n-\frac{1}{2}(M_0M_1+M_1M_0)$.

By Lemma \ref{lem:expanderapprox}, each graph in $\mathcal{H}$ is a $\lambda$-approximation of the complete graph on the same number of vertices. It follows that $I_{\bar{d}}-\tilde{B}\approx_\lambda I_{\bar{d}}-\tilde{J}$  because the quadratic form of a block diagonal matrix equals the sum of the quadratic forms of its blocks. By Lemma \ref{lem:symmequiv} and the fact that $I_d-\tilde{J}$ is PSD, $I_{\bar{d}}-\tilde{B}$ is also a directed $\lambda$-approximation of $I_{\bar{d}}-\tilde{J}$. So for all vectors $x,y\in\mathbb{R}^{\bar{d}}$ we have 
\begin{align*}
\left|x^{T}(\tilde{B}-\tilde{J})y\right| &\leq \frac{\lambda}{2}\cdot (x^{T}(I_{\bar{d}}-\tilde{J})x+y^{T}(I_{\bar{d}}-\tilde{J})y)\\
&\leq \frac{\lambda}{2}\cdot (x^{T}x + y^Ty - 2x^{T}\tilde{J}y).
\end{align*}
The first inequality uses Lemma \ref{lem:dirapprox}. We can add the absolute values on the left-hand side since the right-hand side is always nonnegative ($I_d-\tilde{J}$ is PSD) and invariant to swapping $x$ with $-x$. The second inequality follows from the fact that $\tilde{J}$ is PSD and so 
\[
0 \leq (x-y)^{T}\tilde{J}(x-y) = x^{T}\tilde{J}x+y^{T}\tilde{J}y-2\cdot x^{T}\tilde{J}y.
\]
Fix $v\in\mathbb{R}^{\bar{d}}$ and set $x=R_0v$ and $y=R_1v$. Recall that $R_0$ and $R_1$ are symmetric permutation matrices and hence $R_0^2=R_1^2=I_{\bar{d}}$. Also note that for all square matrices $A$ and vectors $x$, $x^{T}Ax=x^{T}(A+A^{T})x/2$. Combining these observations with the above gives
\begin{align*}
\left|v^{T}\left(\frac{1}{2}\cdot\left(R_0(\tilde{B}-\tilde{J})R_1+R_1(\tilde{B}-\tilde{J})R_0\right)\right)v\right|&=\left|v^{T}R_0(\tilde{B}-\tilde{J})R_1v\right| \\
&\leq \frac{\lambda}{2}\cdot (v^{T}R_0^2v + v^TR_1^2v - 2v^{T}R_0\tilde{J}R_1 v)\\
&=\lambda\cdot (v^{T}v-v^{T}R_0\tilde{J}R_1 v)\\
&=\lambda\cdot v^{T}\left(I-\frac{1}{2}\cdot\left(R_0\tilde{J}R_1+R_1\tilde{J}R_0\right)\right)v
\end{align*}
Rearranging the above shows that
\[
I_{\bar{d}}-\frac{1}{2}\cdot(R_0\bar{B}R_1+R_1\bar{B}R_0)\approx_{\lambda} I_{\bar{d}}-\frac{1}{2}\cdot(R_0\bar{J}R_1+R_1\bar{J}R_0),
\]
which proves the theorem.
\end{proof}

Note that for a graph $G$ with normalized Laplacian $I-M$ and transition matrix $T$, approximating $I-\frac{1}{2}\cdot(M_0M_1+M_1M_0)$ as in Theorem \ref{thm:dprapprox} for $M_0=M^{k_0}$ and $M_1=M^{k_1}$ gives a form of approximation to random walks of length $k_1+k_2$ on $G$, as 
\begin{align*}
I-T^{k_1+k_2} &= D^{1/2}(I-M^{k_1+k_2})D^{-1/2} \\
&=I-\frac{1}{2}\cdot D^{1/2}(M_0M_1+M_1M_0)D^{-1/2}.
\end{align*}

To apply the derandomized product, we need a strongly explicit expander family $\mathcal{H}$ with sizes equal to all of the vertex degrees. However, the commonly used existing constructions of strongly explicit expander families only give graphs of sizes that are subsets of $\mathbb{N}$ such as all powers of 2 or all perfect squares.  In \cite{RozenmanVa05, derandbeyond2017} this was handled by adding self loops to make the vertex degrees all equal and matching the sizes of expanders in explicit families. Adding self loops was acceptable in those works because it does not affect connectivity (the focus of \cite{RozenmanVa05}) or the Laplacian (the focus of \cite{derandbeyond2017}). However it does affect long random walks (our focus), so we cannot add self loops. It was folklore that one can convert any sufficiently dense family of expanders into a family containing expanders of every size (\cite{alon2008elementary}, \cite{goldreich08}) but there was no analysis in terms of spectral expansion or explicitness. We give a self-contained description and proof of such a reduction, with an analysis in terms of both strong explicitness and spectral expansion. Subsequent to this work, Goldreich \cite{gnote} has provided an analysis of a different construction in terms of vertex expansion and strong explicitness, and Alon \cite{alon_inprep} has given a stronger reduction that yields nearly Ramanujan expanders (i.e. ones with nearly optimal spectral expansion as a function of the degree) of every size $n$.

\begin{theorem}
\label{thm:expanders}
There exists a family of strongly explicit expanders $\mathcal{H}$ such that for all $n>1$ and $\lambda\in(0,1)$ there is a $c=\mathrm{poly}(1/\lambda)$ and a $c$-regular graph $H_{n,c}\in \mathcal{H}$ on $n$ vertices with $\lambda(H_{n,c})\leq \lambda$.
\end{theorem}

\begin{proof}
Let $H'$ be a $c'$-regular expander on $m$ vertices such that $n\leq m\leq 2n$, $c'$ is a constant independent of $n$ and $\lambda(H')\leq \lambda'<1/4$. $H'$ can be constructed using already known strongly explicit constructions such as \cite{GabberGa81,ReingoldVaWi01} followed by squaring the graph a constant number of times to achieve $\lambda'<1/4$. We will construct $H$ as follows: Pair off the first $(m-n)$ vertices with the last $(m-n)$ vertices in $H'$ and merge each pair into a single vertex (which will then have degree $2\cdot c'$). To make the graph regular, add $c'$ self loops to all of the unpaired vertices. More precisely, given $u'\in[n]$ and $i'\in[c]=[2\cdot c']$ we compute Rot$_{H}(u',i')$ as follows:
\begin{enumerate}
\item If $1\leq u'\leq m-n$ [$u'$ is a paired vertex]:
\begin{enumerate}
\item If $1\leq i'\leq c'$, let $u=u'$, $i=i'$ [$u'$ is the first vertex in pair]
\item else let $u=m-u'$, $i=i'-c'$ [$u'$ is the second vertex in pair]
\item let $(v,j)=\mathrm{Rot}_{H'}(u,i)$
\end{enumerate}
\item else (if $m-n< u'\leq n$ ) [$u'$ is an unpaired vertex]
\begin{enumerate}
\item If $1\leq i'\leq c'$, let $u=u'$, $i=i'$, and $(v,j)=\mathrm{Rot}_H(u,j)$ [original edge]
\item else let $(v,j)=(u',i')$ [new self loop]
\end{enumerate}
\item
\begin{enumerate}
\item If $v\leq n$, let $(v',j')= (v,j)$
\item else let $v'=m-v$, $j'=j+c'$.
\end{enumerate}
\item Output $(v',j')$
\end{enumerate}

Next we show that $\lambda(H)$ is bounded below 1 by a constant. The theorem then follows by taking the $O(\log 1/\lambda)$th power to drive $\lambda(H)$ below $\lambda$. This gives the graph degree $\mathrm{poly}(1/\lambda)$. 

Let $A'$ be the adjacency matrix of $H'$ and $K'$ be the $m\times m$ all ones matrix. Since $\lambda(H')\leq\lambda'$, Lemma \ref{lem:expanderapprox} implies that 
\[
\frac{1}{c'}\cdot (c'\cdot I- A')\approx_{\lambda'}\frac{1}{m}\cdot(m\cdot I-K').
\]
Define $B$ to be the $m\times n$ matrix such that $B_{u',u}=1$ if and only if vertex $u'\in V(H')$ was merged into vertex $u\in V(H)$ or vertex $u\in V(H')$ was not merged and is labeled vertex $u'$ in $H$. That is, $B_{u',u}=1$ if and only if $u=u'$ or $n\leq u = m-u'$. Then the unnormalized Laplacian of the expander after the merging step is $B^{T}(c'\cdot I-A')B$. Adding self loops to a graph does not change its Laplacian. So applying Proposition \ref{prop:psdfacts} parts 3 and 6 we get
\[
L(H) = \frac{1}{2c'}\cdot B^{T}(c'\cdot I-A')B \approx_{\lambda'}\frac{1}{2m}\cdot B^{T}(m\cdot I-K)B
\]
Note that the righthand side is the normalized Laplacian of the graph $U$ that results from starting with the complete graph on $m$ vertices, merging the same pairs of vertices that are merged in $H$ and adding $m$ self loops to all of the unmerged vertices for regularity.

We finish the proof by showing that $\lambda(U)\leq 1/2$ and thus $H$ is a $(\lambda'+1/2+\lambda'/2)$-approximation of the complete graph by Proposition \ref{prop:psdfacts} Part 2 and Lemma \ref{lem:expanderapprox}. Recalling that $\lambda'<1/4$ completes the proof.

$U$ has at least $m$ edges between every pair of vertices so we can write its transition matrix $T_u$ as 

\[
T_u = \frac{1}{2}\cdot J_m + \frac{1}{2}\cdot E
\]
where $J_m$ is the transition matrix of the complete graph on $m$ vertices with self loops on every vertex and $E$ is the transition matrix for an $m$-regular multigraph. Since the uniform distribution is stationary for all regular graphs, $\vec{1}$ is an eigenvector of eigenvalue 1 for $T_u,J_m,$ and $E$. Thus 
\begin{align*}
\lambda(U)&=\sup_{v\perp\vec{1}}\frac{\|T_uv\|}{\|v\|}\\
&\leq \sup_{v\perp\vec{1}}\frac{\frac{1}{2}\cdot(\|J_mv\|+\|Ev\|)}{\|v\|}\\
&\leq \frac{1}{2}\cdot 0 + \frac{1}{2}\cdot 1,
\end{align*}
which completes the proof.
\end{proof}

\section{Main Result}
\label{sect:mainresult}
In this section we prove Theorem \ref{thm:main}, our main result regarding space bounded computation of the normalized Laplacian of the $r$-step random walk on $G$.

\begin{namedtheorem}[Theorem \ref{thm:main} (restated)]
There is a deterministic algorithm that given an undirected multigraph $G$ on $n$ vertices with normalized Laplacian $I-M$, a nonnegative integer $r$, and $\epsilon>0$, constructs an undirected multigraph $\tilde{G}$ whose normalized Laplacian $\tilde{L}$ is an $\epsilon$-spectral approximation of $L=I-M^r$. That is, for all vectors $v\in\mathbb{R}^n$
\[
(1-\epsilon)\cdot v^{T}Lv \leq v^{T}\tilde{L}v \leq (1+\epsilon)\cdot v^{T}Lv.
\]
The algorithm runs in space $\SC$, where $N$ is the bit length of the input graph $G$. 
\end{namedtheorem}

The algorithm described below is inspired by techniques used in \cite{cheng2015} to approximate random walks with a randomized algorithm in nearly linear time. Our analyses use ideas from the work of Cohen, Kelner, Peebles, Peng, Rao, Sidford, and Vladu on \emph{directed} Laplacian system solvers even though all of the graphs we work with are undirected.
\subsection{Algorithm Description and Proof Overview}
\label{sect:description}
Let $I-M$ be the normalized Laplacian of our input and $r$ be the target power. We will first describe an algorithm for computing $I-M^r$ without regard for space complexity and then convert it into a space-efficient approximation algorithm. The algorithm iteratively approximates larger and larger powers of $M$. On a given iteration, we will have computed $I-M^k$ for some $k<r$ and we use the following operations to increase $k$:
\begin{itemize}
\item Square: $I-M^k \rightarrow I-M^{2k}$,
\item Plus one: $I-M^{k} \rightarrow I-\frac{1}{2}\cdot (M\cdot M^{k}+ M^{k}\cdot M) = I-M^{k+1}$.
\end{itemize}

Interleaving these two operations appropriately can produce any power $r$ of $M$, invoking each operation at most $\log_2 r$ times. To see this, let $b_zb_{z-1}\ldots b_0$ be the bits of $r$ in its binary representation where $b_0$ is the least significant bit and $b_{z}=1$ is the most significant. We are given $I-M=I-M^{b_z}$. The algorithm will have $z$ iterations and each one will add one more bit from most significant to least significant to the binary representation of the exponent. So after iteration $i$ we will have $I-M^{b_zb_{z-1}\ldots b_{z-i}}$.

For iterations $1,\ldots, z$, we read the bits of $r$ from $b_{z-1}$ to $b_{0}$ one at a time. On each iteration we start with some power $I-M^k$. If the corresponding bit is a 0, we square to create $I-M^{2k}$ (which adds a 0 to the binary representation of the current exponent) and proceed to the next iteration. If the corresponding bit is a 1, we square and then invoke a plus one operation to produce $I-M^{2k+1}$ (which adds a 1 to the binary representation of the current exponent). After iteration $z$ we will have $I-M^r$.

Implemented recursively, this algorithm has $\log_2 r$ levels of recursion and uses $O(\log N)$ space at each level for the matrix multiplications, where $N$ is the bit length of the input graph. This results in total space $O(\log r\cdot \log N)$, which is more than we want to use (cf. Theorem \ref{thm:main}). We reduce the space complexity by replacing each square and plus one operation with the corresponding derandomized product, discussed in Section \ref{sect:expanders}. 

Theorem \ref{thm:dprapprox} says that the derandomized product produces spectral approximations to the square and the plus one operation. Since we apply these operations repeatedly on successive approximations, we need to maintain our ultimate approximation to a power of $I-M$. In other words, we need to show that given $\tilde{G}$ such that $I-\tilde{M}\approx_{\epsilon}I-M^k$ we have:
\begin{enumerate}
\item $I-\tilde{M}^2 \approx_{\epsilon} I-M^{2k}$
\item $I-\frac{1}{2}\cdot(M\tilde{M}+\tilde{M}M)\approx_{\epsilon} I-M^{k+1}$.
\end{enumerate}
We prove these in Lemmas \ref{lem:sq} and \ref{lem:plusapprox}. The transitive property of spectral approximation (Proposition \ref{prop:psdfacts} Part 2) will then complete the proof of spectral approximation. 

We only know how to prove items 1 and 2 when $M^k$ is PSD. This is problematic because $M$ is not guaranteed to be PSD for arbitrary graphs and so $M^{k}$ may only be PSD when $k$ is even. Simple solutions like adding self loops (to make the random walk lazy) are not available to us because loops may affect the random walk behavior in unpredictable ways. Another attempt would be to replace the plus one operation in the algorithm with a ``plus two'' operation
\begin{itemize}
\item Plus two: $I-M^{k} \rightarrow I-\frac{1}{2}\cdot (M^2\cdot M^{k}+ M^{k}\cdot M^2) = I-M^{k+2}$.
\end{itemize}
Interleaving the square and plus two would preserve the positive semidefiniteness of the matrix we're approximating and can produce any even power of $M$. If $r$ is odd, we could finish with one plus one operation, which will produce a spectral approximation because $I-M^{r-1}$ is PSD. A problem with this approach is that the derandomized product is defined only for unweighted multigraphs and $M^2$ may not correspond to an unweighted multigraph when $G$ is irregular. (When $G$ is regular, the graph $G^2$ consisting of paths of length 2 in $G$ does have normalized Laplacian $I-M^2$.)

For this reason we begin the algorithm by constructing an unweighted multigraph $G_0$ whose normalized Laplacian $I-M_0$ approximates $I-M^2$ and where $M_0$ is PSD. We can then approximate any power $I-M_0^{r'}$ using the square and plus one operation and hence can approximate $I-M^r$ for any even $r$ (see Lemma \ref{lem:allpowersapprox}). For odd powers, we again can finish with a single plus one operation. 

Our main algorithm is presented below. Our input is an undirected two-way labeled multigraph $G$ with normalized Laplacian $I-M$, $\epsilon\in(0,1)$, and $r=b_zb_{z-1}\ldots b_1b_0$.

\begin{myalg}
\label{alg:main}
\centering
\fbox{
\parbox{0.9\textwidth}{
\textbf{Algorithm \themyalgctr}\\
Input: $G$ with normalized Laplacian $I-M$, $\epsilon\in(0,1)$, $r=b_zb_{z-1}\ldots b_1b_0$  \\
Output: $G_z$ with normalized Laplacian $I-M_z$ such that $I-M_z\approx_{\epsilon}I-M^r$
\begin{enumerate}
\item Set $\mu=\epsilon/(32\cdot z)$
\item Let $\mathcal{H}$ be family of expanders of every size such that $\lambda(H)\leq \mu$ for all $H\in\mathcal{H}$. 
\item Construct $G_0$ such that $I-M_0\approx_{\epsilon/(16\cdot z)}I-M^{2}$ and $M_0$ is PSD.
\item For $i$ in $\{1,\ldots, z-1\}$
\begin{enumerate}
\item If $b_{z-i}=0$, $G_i=G_{i-1}\dpr G_{i-1}$
\item Else $G_i=(G_{i-1}\dpr G_{i-1})\dpr G_0$ 
\end{enumerate}
\item If $b_0=0$ ($r$ even), $G_z=G_{z-1}$
\item Else ($r$ is odd), $G_z=G_{z-1}\dpr G$
\item Output $G_z$
\end{enumerate}
}
}
\end{myalg}
Note that each derandomized product multiplies every vertex degree by a factor of $2\cdot c$. So the degrees of $G,G_0,\ldots, G_z$ are all proportional to one another and the derandomized products in Algorithm \ref{alg:main} are well-defined.

\subsection{Proof of Main Result}
\label{sect:evenpf}
In this section we prove Theorem \ref{thm:main} by showing that Algorithm \ref{alg:main} yields a spectral approximation of our target power $I-M^r$ and can be implemented space-efficiently. First we show that our two operations, square and plus one, preserve spectral approximation.

\begin{lemma}[Adapted from \cite{cheng2015}]
\label{lem:sq}
Let $N$ and $\tilde{N}$ be symmetric matrices such that $I-\tilde{N}\approx_{\epsilon} I-N$ and $N$ is PSD, then $I-\tilde{N}^2 \approx_{\epsilon} I-N^2$.
\end{lemma}
A proof of Lemma \ref{lem:sq} can be found in Appendix \ref{app:sq_preserves_approx}.  Next we show that the plus one operation in our algorithm also preserves spectral approximation.

\begin{lemma}
\label{lem:plusapprox}
Let $\tilde{N}$, $N_1$, and $N_2$ be symmetric matrices with spectral norm at most 1 and suppose that $N_1$ is PSD and commutes with $N_2$. If $I-\tilde{N}\approx_{\epsilon} I-N_1$ then 
\[
I-\frac{1}{2}\cdot (\tilde{N}N_2 + N_2\tilde{N})\approx_{\epsilon} I-N_2N_1.
\]
\end{lemma}
\begin{proof}
Since $I-\tilde{N}\approx_{\epsilon} I-N_1$ and $I-N_1$ is PSD (since $\|N_1\|\leq 1$), Lemmas \ref{lem:dirapprox} and \ref{lem:symmequiv} say that for all $x,y\in \mathbb{R}^n$,
 \[
x^T(\tilde{N}-N_1)y\leq \frac{\epsilon}{2}\cdot (x^T(I-N_1)x +y^T(I-N_1)y).
\]
It follows (replacing $x$ with $N_2x$) that for all vectors $x,y\in \mathbb{R}^n$,
 \[
x^T(N_2\tilde{N}-N_2N_1)y\leq \frac{\epsilon}{2}\cdot (x^TN_2(I-N_1)N_2x +y^T(I-N_1)y).
\]
We now claim that $N_2(I-N_1)N_2\preceq I-N_1\preceq I - N_2N_1$. If the claim is true then the above implies that for all vectors $x,y\in \mathbb{R}^n$
 \[
x^T(N_2\tilde{N}-N_2N_1)y\leq \frac{\epsilon}{2}\cdot (x^T(I-N_2N_1)x +y^T(I-N_2N_1)y),
\]
which by Lemma \ref{lem:dirapprox} implies 
\[
\left\|(I-N_2N_1)^{\dagger/2}(N_2\tilde{N}-N_2N_1)(I-N_2N_1)^{\dagger/2}\right\| \leq \epsilon.
\]
Since for all matrices $U$ it is the case that $\|U+U^T\|/2\leq \|U\|/2 +\|U^T\|/2 =\|U\|$ we have
\[
\left\|(I-N_2N_1)^{\dagger/2}\left(I-\frac{1}{2}\cdot(\tilde{N}N_2 + N_2\tilde{N})-(I-N_2N_1)\right)(I-N_2N_1)^{\dagger/2}\right\| \leq \epsilon.
\]
Therefore $I-\frac{1}{2}\cdot(\tilde{N}N_2 + N_2\tilde{N})$ is a directed $\epsilon$-approximation of $I-N_2N_1$. Since these matrices are symmetric ($N_2$ and $N_1$ commute), Lemma \ref{lem:symmequiv} says that the approximation also holds in the undirected sense. 

It remains to show that $N_2(I-N_1)N_2\preceq I-N_1\preceq I - N_2N_1$. Since $N_1$ and $N_2$ commute, they have the same eigenvectors \cite{horn1990matrix}. So the inequalities reduce to inequalities of their corresponding eigenvalues. Let $\lambda_1,\lambda_2$  be eigenvalues corresponding to the same eigenvector of $N_1,N_2$, respectively. Since $\|N_1\|,\|N_2\|\leq 1$, we have $1-\lambda_1\geq 0$ and $|\lambda_2|\leq 1$, so 
\[
\lambda_2\cdot(1-\lambda_1)\cdot\lambda_2\leq 1-\lambda_1.
\]
This proves that $N_2(I-N_1)N_2\preceq I-N_1$. Also, since $N_1$ is PSD we have $\lambda_1\geq 0$ and thus
\[
1-\lambda_1 \leq 1-\lambda_2\cdot\lambda_1
\]
because $|\lambda_2|\leq 1$. Thus $I-N_1\preceq I-N_2N_1$ and the proof is complete.
\end{proof}
Setting $N_1=M^k$ and $N_2=M$ in Lemma \ref{lem:plusapprox} shows that the plus one operation preserves spectral approximation whenever $M^k$ is PSD. Recall that the first step in Algorithm \ref{alg:main} is to construct a graph $G_0$ with normalized Laplacian $I-M_0$ such that $M_0$ is PSD and $I-M_0$ approximates $I-M^2$. We can then approximate $I-M_0^{k}$ for any $k$ using squaring and plus one because $M_0^k$ will always be PSD. The following Lemma says that $I-M_0^k$ spectrally approximates $I-M^{2k}$.

\begin{lemma}
\label{lem:allpowersapprox}
Let $r$ be a positive integer with bit length $\ell(r)$ and $A$ and $B$ be symmetric PSD matrices with $\|A\|,\|B\|\leq 1$ such that $I-A\approx_{\epsilon}I-B$ and $I-B\approx_{\epsilon}I-A$ for $\epsilon\leq 1/(2\cdot \ell(r))$. Then $I-A^r\approx_{2\cdot\epsilon\cdot \ell(r)}I-B^r$. 
\end{lemma}
\begin{proof} 
Let $b_zb_{z-1}\ldots b_0$ be the bits in the binary representation of $r$ with $b_z=1$ being the most significant bit and $b_0$ being the least. For each $i\in \{0,\ldots,z\}$, let $r_i$ be the integer with binary representation $b_zb_{z-1}\ldots b_{z-i}$. 

The proof is by induction on $i$. As a base case, we have $I-A^{b_z}\approx_{\epsilon}I-B^{b_z}$. We will show for the inductive step that for each $i\in\{1,\ldots,z\}$, if $I-A^{r_{i-1}}\approx_{2\cdot (i-1)\cdot\epsilon}I-B^{r_{i-1}}$ then $I-A^{r_{i}}\approx_{2\cdot i\cdot\epsilon}I-B^{r_{i}}$. Assume $I-A^{r_{i-1}}\approx_{2\cdot (i-1)\cdot\epsilon}I-B^{r_{i-1}}$. By Lemma \ref{lem:sq}, we have 

\[
I-A^{2\cdot r_{i-1}}\approx_{2\cdot (i-1)\cdot\epsilon}I-B^{2\cdot r_{i-1}}
\]
If $b_i=0$ then $2\cdot r_{i-1}=r_i$ so $I-A^{r_i}\approx_{2\cdot (i-1)\cdot\epsilon}I-B^{r_i}$. If $b_i=1$, then by applying Lemma \ref{lem:plusapprox} we have 
\[
I-\frac{1}{2}(A^{2\cdot r_{i-1}}B+BA^{2\cdot r_{i-1}})\approx_{2\cdot (i-1)\cdot\epsilon}I-B^{2\cdot r_{i-1}+1}=I-B^{r_i}.
\]
Applying Lemma \ref{lem:plusapprox} again (recalling that $I-B\approx_{\epsilon}I-A$ and $A$ is PSD) gives
\[
I-\frac{1}{2}(A^{2\cdot r_{i-1}}B+BA^{2\cdot r_{i-1}})\approx_{\epsilon}I-A^{2\cdot r_{i-1}+1}=I-A^{r_i}.
\]
By Proposition \ref{prop:psdfacts} Part 1, this implies 
\[
I-A^{r_i}\approx_{\epsilon/(1-\epsilon)} I-\frac{1}{2}(A^{2\cdot r_{i-1}}B+BA^{2\cdot r_{i-1}})
\]
Putting this together with Proposition \ref{prop:psdfacts} Part 2 we have that $I-A^{r_i}$ is a $(\epsilon/(1-\epsilon)+2\cdot (i-1)\cdot\epsilon+2\cdot (i-1)\cdot\epsilon^2/(1-\epsilon))$-approximation of $I-B^{r_i}$. Note that
\begin{align*}
2\cdot (i-1)\cdot\epsilon+\frac{\epsilon+2\cdot (i-1)\cdot\epsilon^2}{1-\epsilon}&\leq 2\cdot (i-1)\cdot\epsilon+\frac{2\cdot\epsilon-2\cdot\epsilon^2}{1-\epsilon}\\
&=2\cdot i \cdot \epsilon
\end{align*}
where the first inequality follows from our assumption that $\epsilon\leq 1/(2\cdot \ell(r))$. So $I-A^{r_i} \approx_{2\cdot i\cdot \epsilon}I-B^{r_i}$, completing the inductive step. This shows that when $i=\ell(r)$, $I-A^{r}_{2\cdot \ell(r)\cdot \epsilon}I-B^{r}$, as desired. 
\end{proof}

Now we can prove Theorem \ref{thm:main}. We prove the theorem with three lemmas: Lemma \ref{lem:buildpsdG} shows how to construct the graph $G_0$ needed in Algorithm \ref{alg:main}, Lemma \ref{lem:pf_of_spec_approx} argues that the algorithm produces a spectral approximation to $I-M^r$, and Lemma \ref{lem:pf_of_space} shows that the algorithm can be implemented in space $\SC$. 
\subsubsection{Building \texorpdfstring{$G_0$}{}}
\begin{lemma}
\label{lem:buildpsdG}
There is an algorithm that takes an undirected, unweighted multigraph $G$ with normalized Laplacian $I-M$ and a parameter $\epsilon>0$, and outputs a rotation map Rot$_{G_0}$ for an undirected, unweighted multigraph $G_0$ with a two-way labeling and normalized Laplacian $I-M_0$ such that:
\begin{enumerate}
\item $M_0$ is PSD,
\item $I-M_0\approx_{\epsilon}I-M^2$,
\item The algorithm uses space $O(\log N +\log(1/\epsilon))$, where $N$ is the bit length of the input graph $G$.
\end{enumerate}
\end{lemma}
\begin{proof}
Let $\delta=1/\lceil{4/\epsilon}\rceil$ and $t=1/\delta$, an integer. Let $\mathcal{H}$ be a family of $c$-regular expanders of every size from Theorem \ref{thm:expanders}, such that for every $H\in\mathcal{H}$, $\lambda(H)\leq \delta$ (and hence $c=\mathrm{poly}(1/\delta)$). 

Let $\tilde{G}=G\dpr G$ be the derandomized square of $G$ with normalized Laplacian $I-\tilde{M}$. Each vertex $v$ in $\tilde{G}$ has degree $\tilde{d_v}=2\cdot c\cdot d_v$, where $d_v$ is the degree of $v$ in $G$. We construct $G_0$ as follows: duplicate every edge of $\tilde{G}$ to have multiplicity $t$ and then for each vertex $v$, add $\tilde{d_v}$ self loops. So for each vertex $v$ in $G_0$, $v$ has degree $(t+1)\cdot 2\cdot c\cdot d_v$ and hence $G_0$ has the same stationary distribution as $G$. Note that we can write 
\[
M_0 = (t\cdot \tilde{M} + I)/(t+1).
\]
First we show that $M_0$ is PSD. From Theorem \ref{thm:dprapprox}, we have $I-\tilde{M}\approx_{\delta}I-M^2$, so $I-\tilde{M}\preceq (1+\delta)\cdot (I-M^2)\preceq (1+\delta)\cdot I$, since $M^2$ is PSD. Thus $\tilde{M}\succeq -\delta\cdot I$ and
\[
M_0\succeq\frac{t\cdot(-\delta \cdot I)+I}{t+1}\succeq 0.
\]

Next we prove that $I-M_0\approx_{\epsilon}I-M^2$ 
\begin{align*}
I-M_0&=(t/(t+1))\cdot(I-\tilde{M})\\
&=\left(\frac{1}{1+\delta}\right)\cdot(I-\tilde{M})\\
&\preceq I-M^2.
\end{align*}
Observe that since $I-\tilde{M}\approx_{\delta} I-M^2$, we also have
\begin{align*}
I-M_0&=\left(\frac{1}{1+\delta}\right)\cdot(I-\tilde{M})\\
&\succeq \left(\frac{1-\delta}{1+\delta}\right)\cdot(I-M^2)\\
&\succeq (1-\epsilon)\cdot(I-M^2).
\end{align*}

We can construct a two-way labeling of $G$ in space $O(\log N)$ by arbitrarily numbering the edges incident to each vertex. Computing Rot$_{\tilde{G}}$ involves computing Rot$_G$ twice and the rotation map of an expander in $\mathcal{H}$ once. For a given vertex degree $d$ in $G$, Rot$_{H_{d}}$ can be computed in space $O(\log (d\cdot c))=O(\log N + \log(1/\epsilon))$. Duplicating the edges and adding self loops for Rot$_{G_0}$ adds at most $O(\log N+\log(1/\epsilon))$ overhead for a total of $O(\log N+\log(1/\epsilon))$ space.
\end{proof}

\subsubsection{Proof of Spectral Approximation}
\begin{lemma}
\label{lem:pf_of_spec_approx}
Let $G$ be an undirected multigraph with normalized Laplacian $I-M$, $r$ be a positive integer and $\epsilon\in(0,1)$. Let $G_z$ be the output of Algorithm \ref{alg:main} with normalized Laplacian $I-M_z$. Then
\[
I-M_z \approx_{\epsilon} I-M^r
\]
\end{lemma}
\begin{proof}
Let $b_zb_{z-1}\ldots b1 b_0$ be the binary representation of $r$. Recall that for the derandomized products in our algorithm we use a family of $c$-regular expanders $\mathcal{H}$ from Theorem \ref{thm:expanders} such that for every $H\in\mathcal{H}$, $\lambda(H)\leq\mu=\epsilon/(32\cdot z)$ (and hence $c=\mathrm{poly}(1/\mu)=\mathrm{poly}((\log r)/\epsilon)$).

We construct $G_0$ with normalized Laplacian $I-M_0$ as in Lemma \ref{lem:buildpsdG} such that $M_0$ is PSD and $I-M_0 \approx_{\epsilon/(16\cdot z)}I-M^2$. By Proposition \ref{prop:psdfacts} Part 1, and the fact that 
\begin{align*}
\frac{\epsilon/(16\cdot z)}{1-\epsilon/(16\cdot z)}&=\frac{\epsilon}{(16\cdot z)-\epsilon}\\
&\leq \frac{\epsilon}{8\cdot z},
\end{align*}
we also have $I-M^2 \approx_{\epsilon/(8\cdot z)}I-M_0$.  

For each $i\in\{0,\ldots z\}$ let $r_i$ be the integer with binary representation $b_zb_{z-1}\ldots b_{z-i}$ and let $I-M_i$ be the normalized Laplacian of $G_i$. We will prove by induction on $i$ that $G_i$ is a $(4\cdot\mu\cdot i)$-approximation to $I-M_0^{r_i}$. Thus, $G_{z-1}$ is a $4\cdot\mu\cdot(z-1)\leq \epsilon/8$-approximation to $I-M_0^{r_{z-1}}$.

The base case is trivial since $r_0=1$. For the induction step, suppose that $I-M_{i-1}\approx_{4\cdot\mu\cdot (i-1)}I-M_{0}^{r_{z-i+1}}$. On iteration $i$, if $b_{z-i}=0$, then $G_{i}=G_{i-1}\dpr G_{i-1}$. So we have 
\begin{align*}
I-M_i &\approx_{\mu} I-M_{i-1}^2 \\
&\approx_{4\cdot\mu\cdot (i-1)} I-M_{0}^{2\cdot r_{i-1}}\\
&=I-M_0^{r_{i}}
\end{align*}
where the first approximation uses Theorem \ref{thm:dprapprox} and the second uses Lemma \ref{lem:sq}. By Proposition \ref{prop:psdfacts} Part 2 this implies that $I-M_i$ approximates $I-M_0^{r_{i}}$ with approximation factor 
\[
\mu +4\cdot\mu\cdot (i-1)+4\cdot\mu^2\cdot (i-1)\leq 4\cdot\mu\cdot i
\]
where we used the fact that $\mu<1/(32\cdot (i-1))$.

If $b_{z-i}=1$, $G_{i}=(G_{i-1}\dpr G_{i-1})\dpr G_0$. Let $I-M_{\mathrm{ds}}$ be the normalized Laplacian of $G_{i-1}\dpr G_{i-1}$. By the analysis above, $I-M_{\mathrm{ds}}$ is a $(\mu +4\cdot\mu\cdot (i-1)+4\cdot\mu^2\cdot (i-1))$-approximation of $I-M_0^{2\cdot r_{i-1}}$. By Theorem \ref{thm:dprapprox} and Lemma \ref{lem:plusapprox} we have 
\begin{align*}
I-M_i &\approx_{\mu} I-\frac{1}{2}\cdot(M_{\mathrm{ds}}M_0 + M_0M_{\mathrm{ds}})\\
&\approx_{\mu +4\cdot\mu\cdot (i-1)+4\cdot\mu^2\cdot (i-1)} I-M_0^{2\cdot r_{i-1}}M_0\\
&=I-M_0^{r_i}
\end{align*}
Applying Proposition \ref{prop:psdfacts} Part 2 and noting that $\mu\leq 1/(32\cdot(i-1))$ we get
\[
I-M_i \approx_{4\cdot\mu\cdot i} I-M_0^{r_i}.
\]

So we conclude that $I-M_{z-1}\approx_{\epsilon/8}I-M_0^{r_{z-1}}$. Furthermore, by Lemma \ref{lem:allpowersapprox} we have 
\[
I-M_0^{r_{z-1}}\approx_{\epsilon/8}I-M^{2\cdot r_{z-1}}.
\]
By Proposition \ref{prop:psdfacts} Part 2, and the fact that $\epsilon\leq 1$, this gives
\[
I-M_{z-1}\approx_{\epsilon/3}I-M^{2\cdot r_{z-1}}
\]
If $b_0=0$ then $2\cdot r_{z-1}=r$ and we are done. If $b_0=1$ then we apply one more plus one operation using our original graph $G$ to form $G_z=G_{z-1}\dpr G$ such that
\begin{align*}
I-M_z &\approx_{\mu} I-\frac{1}{2}\cdot (M_{z-1}M+MM_{z-1})\\
&\approx_{\epsilon/3}I-M^{2\cdot r_{z-1}+1}\\
&=I-M^{r}.
\end{align*}
Applying Proposition \ref{prop:psdfacts} Part 2 then gives $I-M_z \approx_{\epsilon} I-M^{r}$.
\end{proof}

\subsubsection{Analysis of Space Complexity}
\begin{lemma}
\label{lem:pf_of_space}
Algorithm \ref{alg:main} can be implemented so that given an undirected multigraph $G$, a positive integer $r$, and $\epsilon\in(0,1)$, it computes its output $G_z$ in space $\SC$, where $N$ is the bit length of the input graph $G$. 
\end{lemma}
\begin{proof}
We show how to compute Rot$_{G_z}$ in space $\SC$. Let $b_zb_{z-1}\ldots b_0$ be the binary representation of $r$. Following Algorithm \ref{alg:main}, $G_0$ is constructed with normalized Laplacian $I-M_0\approx_{\epsilon/(16\cdot z)}I-M^2$. From Lemma \ref{lem:buildpsdG}, we know Rot$_{G_0}$ can be computed in space $O(\log N + \log(16\cdot z/\epsilon))=O(\log N+\log(1/\epsilon)+\log\log r)$. Let $d_1,\ldots,d_n$ be the vertex degrees in $G_0$ and $\dmax$ be the maximum degree. 

The algorithm is presented to have $z$ iterations, where on iteration $i\in[z-1]$, if $b_{z-i}=0$ the derandomized product is invoked once, and if $b_{z-i}=1$, it is invoked twice. On iteration $z$ it is either invoked once ($b_0=1$)  or not at all ($b_0=0$). It will be simpler for us to think of each derandomized product happening in its own iteration. So we will consider $\tau=z+w=O(\log r)$ iterations where $w$ is the number of ones in $b_{z-1},\ldots,b_{0}$. On iterations $1,\ldots,z-1$, there are $z-1$ derandomized square operations and $w$ plus one operations. The final iteration will either have a plus one operation with the graph $G$ (if $b_0=1$) or no operation.

We copy the bits of $r$ into memory and expand them into $\tau$ bits as follows: for $i\in\{1,\ldots z-1\}$ if $b_{z-i}=0$, record a 0 (corresponding to a derandomized square) and if $b_{z-i}=1$, record a 0 followed by a 1 (corresponding to a derandomized square followed by a plus one operation). Finish by just recording $b_z$ at the end. Now we have $\tau$ bits $t_1,\ldots, t_\tau$ in memory where for $i<\tau$, $t_i=0$ if the $i$th derandomized product in our algorithm is a derandomized square and $t_i=1$ if the $i$th derandomized product is a plus one with the graph $G_0$. If $t_\tau=0$, we do no derandomized product on the last iteration and if $t_\tau=1$ we apply the plus one operation using $G$ instead of $G_0$ as described in the algorithm. 

We also re-number our graphs to be $G_{1},\ldots, G_{\tau}$ where $G_i$ is the graph produced by following the derandomized products corresponding to $t_1,\ldots, t_i$. For each $i\in[\tau]$ and $v\in[n]$, vertex $v$ in graph $G_i$ has degree $(2\cdot c)^{i}\cdot d_v$ because each derandomized product multiplies every vertex degree by a factor of $2\cdot c$. 

Since our graphs can be irregular, the input to a rotation map may have a different length than its output. To simplify the space complexity analysis, when calling a rotation map, we will pad the edge labels to always have the same length as inputs and outputs to the rotation map. For each graph $G_i$, we pad its edge labels to have length $\ell_i=\lceil{\log_2\dmax}\rceil +i\cdot \lceil{\log_2(2\cdot c)}\rceil$. 

Sublogarithmic-space complexity can depend on the model, so we will be explicit about the model we use. We compute the rotation map of each graph $G_i$ on a multi-tape Turing machine with the following input/output conventions:

\begin{itemize}
\item Input Description:
\begin{itemize}
\item Tape 1 (read-only): Contains the input $G$, $r$, and $\epsilon$ with the head at the leftmost position of the tape.
\item Tape 2 (read-write): Contains the input to the rotation map  $(v_0,k_0)$, where $v_0\in[n]$ is a vertex of $G_i$, and $k_0$ is the label of an edge incident to $v_0$ padded to have total length $\ell_i$. The tapehead is at the rightmost end of $k_0$. The rest of the tape may contain additional data.
\item Tape 3: (read-write) Contains the bits $t_1,\ldots, t_{\tau}$ with the head pointing at $t_i$.
\item Tapes 4+: (read-write): Blank worktapes with the head at the leftmost position.
\end{itemize}

\item Output Description:
\begin{itemize}
\item Tape 1: The head  should be returned to the leftmost position.
\item Tape 2: In place of $(v_0,k_0)$, it should contain $(v_2,k_2)=\mathrm{Rot}_{G_i}(v_0,k_0)$, where $v_2\in[n]$, and $k_2$ is padded to have total length $\ell_i$. The head should be at the rightmost position of $k_2$ and the rest of the tape should remain unchanged from its state at the beginning of the computation.
\item Tape 3: Contains the bits $t_1,\ldots, t_{\tau}$ with the head pointing at $t_i$.
\item Tapes 4+: (read-write): Are returned to the blank state with the heads at the leftmost position.
\end{itemize}
\end{itemize}

Let $\Space(G_i)$ be the space used on tapes other than tape 1 to compute Rot$_{G_i}$. We will show that $\Space(G_i)=\Space(G_{i-1}) + O(\log c)$. Recalling that $\Space(G_0)=O(\log N+\log(1/\epsilon)+\log\log r)$ and unraveling the recursion gives 
\begin{align*}
\Space(G_{z})&=O(\log N +\log(1/\epsilon)+\log\log r+ \tau\cdot \log c)\\
&=O(\log N + \log(1/\epsilon)+\log\log r+ \log r\cdot\log(\mathrm{poly}(\log r)/\epsilon))\\
&= O(\log N + (\log r)\cdot\log(1/\epsilon)+(\log r)\cdot \log\log r)
\end{align*}
as desired. Now we prove the recurrence on $\Space(G_i)$. We begin with $(v_0, k_0)$ on Tape 2 (possibly with additional data) and the tapehead at the far right of $k_0$. We parse $k_0$ into $k_0=(j_0, a_0, b)$ where $j_0$ is an edge label in $[(2\cdot c)^{i-1}\cdot d_{v_0}]$ padded to have length $\ell_{i-1}$, $a_0\in [c]$, and $b\in\{0,1\}$. 

Note that $G_i=G_{i-1}\dpr G'$ where for $i\neq \tau$, we have $G'=G_{i-1}$ if $t_{i-1}=0$ and $G'=G_{0}$ when $t_{i-1}=1$. We compute Rot$_{G_i}$ according to Definition \ref{def:derandproduct}. We move the head left to the rightmost position of $j_0$. If $b=0$, we move the third tapehead to $t_{i-1}$ and recursively compute Rot$_{G_{i-1}}(v_0,j_0)$ so that Tape 2 now contains $(v_1,j_1,a_0, b)$ (with $j_1$ padded to have the same length as $j_0$). The vertex $v_1$ in the graph $G_{i-1}$ has degree $d'=(2\cdot c)^{i-1}\cdot d_{v_1}$ so we next compute Rot$_{H_{d'}}(j_1,a_0)$ so that $(v_1, j_2, a_1, b)$ is on the tape. Finally we compute Rot$_{G'}(v_1,j_2)$ and flip $b$ to finish with $(v_2,j_3,a_1, \bar{b})$ on the second tape. We then move the third tapehead to $t_i$. If $b=1$ then we just swap the roles of $G_{i-1}$ and $G'$ above.  

So computing Rot$_{G_i}$ involves computing the rotation maps of $G_{i-1}$, $H_{d'}$, and $G'$ each once. Note that each of the rotation map evaluations occur in succession and can therefore reuse the same space. Clearly $\Space(G')\leq \Space(G_{i-1})$ because either $G'=G_{i-1}$ or $G'$ is either $G_0$ or $G$, both of whose rotation maps are subroutines in computing Rot$_{G_{i-1}}$. Computing Rot$_{H_{d'}}$ adds an overhead of at most $O(\log c)$ space to store the additional edge label $a_0$ and the bit $b$. So we can compute the rotation map of $G_\tau$ in space $\SC$. 
\end{proof}

\section{Corollaries}
\label{sect:cors}
\subsection{Random Walks}
Our algorithm immediately implies Theorem \ref{thm:conductance}, which we prove below. 

\begin{namedtheorem}[Theorem \ref{thm:conductance} Restated]
There is a deterministic algorithm that given an undirected multigraph $G$ on $n$ vertices, a positive integer $r$, a set of vertices $S$, and $\epsilon>0$, computes a number $\tilde{\Phi}$ such that
\[
(1-\epsilon)\cdot \Phi_r(S) \leq \tilde{\Phi} \leq (1+\epsilon)\cdot \Phi_r(S)
\]
and runs in space $\SC$, where $N$ is the bit length of the input graph $G$. 
\end{namedtheorem}

\begin{proof}[Proof of Theorem \ref{thm:conductance}]
Let $D$ be the diagonal degree matrix and $I-M$ be the normalized Laplacian of $G$. Let $v=D^{1/2}e_S$ where $e_S$ is the characteristic vector of the set $S$. Let $d_S$ be the sum of the degrees of vertices in $S$. Then using the fact that $I-M^r=D^{-1/2}(I-T^r)D^{1/2}$ where $T$ is the transition matrix of $G$ gives:
\begin{align*}
\frac{1}{d_S}\cdot v^{T}(I-M^r)v &= \frac{1}{d_S}\cdot e_S^{T}D^{1/2}D^{-1/2}(I-T^r)D^{1/2}D^{1/2}e_S \\
&=\frac{1}{d_S}\cdot e_S^{T}De_{S} -  e_S^{T}(T^{r}(De_S/d_S))\\
&=1-\Pr[V_r\in S| V_0\in S]\\
&=\Phi_r(S)
\end{align*}
where the penultimate equality follows from the fact that $De_S/d_S$ is the probability distribution over vertices in $S$ where each vertex has mass proportional to its degree, i.e. the probability distribution $V_0\|(V_0\in S)$. Multiplying this distribution by $T^r$ gives the distribution of $V_r\|(V_0\in S)$. Multiplying this resulting distribution on the left by $e_S^{T}$, sums up the probabilities over vertices in $S$, which gives the probability that our random walk ends in $S$. 

From Theorem \ref{thm:main}, we can compute a matrix $\tilde{L}$ such that $\tilde{L}\approx_{\epsilon} I-M^r$ in space $\SC$. It follows from Proposition \ref{prop:psdfacts}, Part 6 and the definition of spectral approximation that 
\[
(1-\epsilon)\cdot \Phi_{r}(S)\leq \frac{1}{d_S}\cdot v^{T}\tilde{L}v\leq (1+\epsilon)\cdot \Phi_{r}(S).
\]
\end{proof}
Our algorithm also implies an algorithm for approximating random walk matrix polynomials.
\begin{definition}[Random walk matrix polynomials]
Let $G$ be an undirected multigraph with Laplacian $D-A$ and let $\alpha$ be a vector of nonnegative scalars $\alpha=(\alpha_1,\ldots,\alpha_{\ell})$ such that $\sum_{i\in[\ell]}\alpha_i=1$. The normalized \emph{$\ell$-degree random walk matrix polynomial} of $G$ with respect to $\alpha$ is defined to be:
\begin{align*}
L_{\alpha}(G)&\coloneqq I-\sum_{r=1}^{\ell}\alpha_{r}\cdot M^r\\
&=D^{1/2}\left(I-\sum_{r=1}^{\ell}\alpha_r\cdot T^r\right)D^{-1/2},
\end{align*}
where $M=D^{-1/2}AD^{-1/2}$ and $T=AD^{-1}$
\end{definition}
Random walk matrix polynomials are Laplacian matrices that arise in algorithms for approximating $q$th roots of symmetric diagonally dominant matrices and efficient sampling from Gaussian graphical models \cite{loh12,cheng2015,ChengCLPT14}. 
\begin{corollary}
There is a deterministic algorithm that given an undirected multigraph $G$, a positive integer $\ell$, a vector of nonnegative scalars $\alpha=(\alpha_1,\ldots,\alpha_{\ell})$ that sum to 1, and $\epsilon>0$, computes an $\epsilon$-approximation to $L_{\alpha}(G)$ and runs in space $\SC$.
\end{corollary}
\begin{proof}
For each $r\in[\ell]$, let $L_r$ be an $\epsilon$-approximation of $I-M^r$. Applying Proposition \ref{prop:psdfacts}, Parts 5 and 6 gives
\begin{align*}
\sum_{r=1}^{\ell}\alpha_r\cdot L_{r} &\approx_{\epsilon}\sum_{r=1}^{\ell}\alpha_{r}\cdot I-M^r\\
&=I-\sum_{r=1}^{\ell}\alpha_{r}\cdot M^r\\
&=L_{\alpha}(G)
\end{align*}
Our algorithm can compute each $L_r$ in space $\SC$. Multiplication and iterated addition can both be computed in $O(\log N)$ space where $N$ is the input bit length \cite{BCH86,ABH02}. So by the composition of space bounded algorithms (Proposition \ref{prop:composition}), $L_{\alpha}(G)$ can be approximated in space $\SC$.
\end{proof}

\subsection{Odd Length Walks in Nearly Linear Time}
Our approach to approximating odd length walks deterministically and space-efficiently also leads to a new result in the context of nearly linear-time (randomized) spectral sparsification algorithms. Specifically, we extend the following Theorem of Cheng, Cheng, Liu, Peng, and Teng \cite{cheng2015}.

\begin{theorem}[\cite{cheng2015}]
\label{thm:lintime}
There is a randomized algorithm that given an undirected weighted graph  $G$ with $n$ vertices, $m$ edges, and normalized Laplacian $I-M$, \emph{even} integer $r$, and $\epsilon>0$ constructs an undirected weighted graph $\tilde{G}$ with normalized Laplacian $\tilde{L}$ containing $O(n\log n/\epsilon^2)$ non-zero entries, in time $O(m\cdot \log^3 n\cdot \log^5 r/\epsilon^4)$, such that $\tilde{L}\approx_{\epsilon}I-M^r$ with high probability.
\end{theorem}
Our approach to approximating odd length walks can be used to extend Theorem \ref{thm:lintime} to odd $r$. 

\begin{corollary}
\label{cor:lintime}
There is a randomized algorithm that given an undirected weighted graph $G$ with $n$ vertices, $m$ edges, and normalized Laplacian $I-M$, \emph{odd} integer $r$, and $\epsilon>0$ constructs an undirected weighted graph $\tilde{G}$ with normalized Laplacian $\tilde{L}$ containing $O(n\log n/\epsilon^2)$ non-zero entries, in time $O(m\cdot \log^3 n\cdot \log^5 r/\epsilon^4)$, such that $\tilde{L}\approx_{\epsilon}I-M^r$ with high probability.
\end{corollary}

Our proof of Corollary \ref{cor:lintime} uses Theorem \ref{thm:lintime} as a black box. So in fact, given $G$ with normalized Laplacian $I-M$ and any graph $\tilde{G}$ whose normalized Laplacian approximates $I-M^r$ for even $r$, we can produce an approximation to $I-M^{r+1}$ in time nearly linear in the sparsities of $G$ and $\tilde{G}$. To prove the corollary, we use the same method used in \cite{PS13} and \cite{CKPPRSV16} for sparsifying two-step walks on undirected and directed graphs, respectively. The idea is that the graphs constructed from two-step walks can be decomposed into the union of \emph{product graphs}: graphs whose adjacency matrices have the form $xy^{T}$ for vectors $x,y\in\mathbb{R}^n$. We use the following fact from \cite{CKPPRSV16} that says that product graphs can be sparsified in time that is nearly-linear in the number of non-zero entries of $x$ and $y$ rather than the number of non-zero entries in $xy^{T}$, which may be much larger.

\begin{lemma}[Adapted from \cite{CKPPRSV16} Lemma 3.18]
\label{lem:sparsifyproduct}
Let $x,y$ be non-negative vectors with $\|x\|_1=\|y\|_1=r$ and let $\epsilon\in(0,1)$. Furthermore, let $s$ denote the total number of non-zero entries in $x$ and $y$ and let $L=\mathrm{diag}(y)-\frac{1}{r}\cdot xy^T$. Then there is an algorithm that in time $O(s\cdot\log s/\epsilon^2)$ computes a matrix $\tilde{L}$ with $O(s\cdot\log s/\epsilon^2)$ non-zeros such that $\tilde{L}$ is a directed $\epsilon$-approximation of $L$ with high probability.
\end{lemma}

After using Lemma \ref{lem:sparsifyproduct} to sparsify each product graph in our decomposition, we then apply an additional round of graph sparsification.

\begin{lemma}[\cite{KPPS16}]
\label{lem:st}
Given an undirected graph $G$ with $n$ vertices, $m$ edges, and Laplacian $L$ and $\epsilon>0$, there is an algorithm that computes a graph $\tilde{G}$ with Laplacian $\tilde{L}$ containing $O(n\cdot\log n/\epsilon^2)$ non-zero entries in time $O(m\cdot \log^2 n/\epsilon^2)$ such that $\tilde{L}\approx_{\epsilon}L$ with high probability. 
\end{lemma}
Now we can prove Corollary \ref{cor:lintime}

\begin{proof}[Proof of Corollary \ref{cor:lintime}]
Theorem \ref{thm:lintime} says that we can compute a graph $\tilde{G}$ with normalized Laplacian $I-\tilde{M}$ with $O(n\log n/\epsilon^2)$ non-zero entries, in time $O(m\cdot \log^3 n\cdot \log^5 r/\epsilon^4)$, such that $I-\tilde{M}\approx_{\epsilon/8}I-M^{r-1}$ with high probability. By Lemma \ref{lem:plusapprox} we have
\begin{equation}
\label{eq:oddtime}
I-\frac{1}{2}\cdot(\tilde{M}M+M\tilde{M}) \approx_{\epsilon/8} I-M^{r}.
\end{equation}

Our goal is to sparsify the lefthand side. Note that since $I-\tilde{M}$ spectrally approximates $I-M^{r-1}$, the corresponding graphs must have the same stationary distribution and hence proportional vertex degrees. In other words there is a number $k$ such that for all vertices $v\in[n]$ we have $\mathrm{deg}_{\tilde{G}}(v)=k\cdot\mathrm{deg}_{G}(v)$. We will think of the graph that adds one step to our walk as $k\cdot G$ rather than $G$ because $k\cdot G$ and $\tilde{G}$ have the same degrees and the normalized Laplacian of $k\cdot G$ is the same as the normalized Laplacian of $G$.

Let $A$ and $\tilde{A}$ be the adjacency matrices of $k\cdot G$ and $\tilde{G}$, respectively and let $D$ be the diagonal matrix of vertex degrees. Let $Q=D-AD^{-1}\tilde{A}$ and note that $Q$ is the Laplacian of a weighted directed graph. We will show how to compute a sparse directed approximation of $Q$ and use this to show how to compute a sparse approximation to the lefthand side of Equation \ref{eq:oddtime}. Our approach is inspired by similar arguments from \cite{PS13, CKPPRSV16}. We decompose $Q$ into $n$ product graphs as follows. For each $i\in[n]$ let
\[
Q_i = \mathrm{diag}(\tilde{A}_{i,:})-\frac{1}{D_{i,i}}\cdot A_{:,i}\tilde{A}_{i,:}^T
\]
where $\tilde{A}_{i,:}$ and $A_{:,i}$ denote the $i$th row of $\tilde{A}$ and the $i$th column of $A$, respectively. Observe that $Q_i$ is a directed Laplacian of a bipartite graph between the neighbors of vertex $i$ in $k\cdot G$ and the neighbors of $i$ in $\tilde{G}$ and that
$Q=\sum_{i\in [n]}Q_i$. Furthermore, each $Q_i$ is a product graph and hence can be sparsified using Lemma \ref{lem:sparsifyproduct}. Set $x_i=A_{:,i}$, $y_i=\tilde{A}_{i,:}$, $r_i=D_{i,i}$, and let $s_i$ be the total number of non-zero entries in $x$ and $y$. Note that $\|x_i\|_1=\|y_i\|_1=r_i$ because $k\cdot G$ and $\tilde{G}$ have the same vertex degrees. By Lemma \ref{lem:sparsifyproduct}, for each $i\in[n]$ we can compute a directed $\epsilon/8$-approximation $\tilde{Q}_i$ of $Q_i$ containing $O(s_i\cdot \log s_i/\epsilon^2)$ entries in time $O(s_i\cdot \log s_i/\epsilon^2)$. Applying the lemma to each $Q_i$ yields $\tilde{Q}=\sum_{i\in [n]}\tilde{Q}_i$, which contains $O(m\cdot \log m/\epsilon^2)$ non-zero entries and can be computed in time $O(m\cdot \log m/\epsilon^2)$ because $\sum_{i\in[n]}s_i=O(m)$. By Lemma \ref{lem:dir_implies_undir} we have 
\[
\frac{1}{2}\cdot(\tilde{Q}_i+\tilde{Q}_i^T)\approx_{\epsilon/8}\frac{1}{2}\cdot(Q_i+Q_i^T)
\]
for all $i\in[n]$ with high probability. It follows from Proposition \ref{prop:psdfacts} Part 5 that
\begin{align*}
\frac{1}{2}\cdot(\tilde{Q}+\tilde{Q}^T)&=\frac{1}{2}\cdot\sum_{i\in [n]}(\tilde{Q}_i+\tilde{Q}_i^T)\\
&\approx_{\epsilon/8}\frac{1}{2}\cdot\sum_{i\in [n]}(Q_i+Q_i^T)\\
&=\frac{1}{2}\cdot(Q+Q^T)
\end{align*}
with high probability. From Proposition \ref{prop:psdfacts} Part 3, we then get
\begin{align*}
D^{-1/2}\frac{1}{2}\cdot(\tilde{Q}+\tilde{Q}^T)D^{-1/2}&\approx_{\epsilon/8}D^{-1/2}\frac{1}{2}\cdot(Q+Q^T)D^{-1/2}\\
&=I-\frac{1}{2}\cdot(\tilde{M}M+M\tilde{M})
\end{align*}
with high probability. Applying Lemma \ref{lem:st} we can re-sparsify the graph corresponding to $D^{-1/2}\frac{1}{2}\cdot(\tilde{Q}+\tilde{Q}^T)D^{-1/2}$ to produce a graph $G'$ whose normalized Laplacian $I-M'$ has $O(n\cdot \log n/\epsilon^2)$ non-zero entries and $I-M'\approx_{\epsilon/8}D^{-1/2}\frac{1}{2}\cdot(\tilde{Q}+\tilde{Q}^T)D^{-1/2}$ with high probability. This takes additional time $O(m\cdot \log^2 n/\epsilon^2)$ due to Theorem 1.1 of \cite{KPPS16}. Applying Proposition \ref{prop:psdfacts} Part 2 twice we get that $I-M'\approx_{\epsilon}I-M^r$ and the total running time for the procedure was $O(m\cdot \log^3 n\cdot \log^5 r/\epsilon^4)$.
\end{proof}

\newpage
\bibliographystyle{alphanum}
\bibliography{graphpowers,pseudorandomness}

\newpage
\appendix
\section{Proof of Lemma \ref{lem:expanderapprox}}
\label{app:expanderapprox}
\begin{namedtheorem}[Lemma \ref{lem:expanderapprox} (restated)]
Let $H$ be a $c$-regular undirected multigraph on $n$ vertices with transition matrix $T$ and let $J\in\mathbb{R}^{n\times n}$ be a matrix with $1/n$ in every entry (i.e. $J$ is the transition matrix of the complete graph with a self loop on every vertex). Then $\lambda(H)\leq \lambda$ if and only if $I-T\approx_{\lambda} I-J$.
\end{namedtheorem}
\begin{proof}
Let $v_1,v_2,\ldots,v_n$ be the orthonormal eigenvectors of $I-T$ where $v_1$ is the uniform distribution (the vector containing $1/n$ in every coordinate). Note that these are also eigenvectors of $I-J$ as $(I-J)v_1=\vec{0}$ and $(I-J)v_i=v_i$ for all $i\neq 1$. 

Let $V$ be a matrix with $v_1,\ldots,v_n$ as columns. $V$ is an orthogonal matrix so $V^TV=I$. From Proposition \ref{prop:psdfacts} Part 3, we have $I-T\approx_{\lambda}I-J$ if and only if $V^T(I-T)V\approx_{\lambda}V^T(I-J)V$. $V^T(I-T)V$ is a diagonal matrix with the eigenvalues $0,1-\lambda_2,\ldots 1-\lambda_n$ of $I-T$ along the diagonal and $V^T(I-J)V$ is a diagonal matrix with the eigenvalues $0,1,\ldots, 1$ of $I-J$ along the diagonal. Diagonal matrices spectrally approximate each other if and only if their diagonal entries satisfy the spectral inequalities, so $V^T(I-T)V\approx_{\lambda}V^T(I-J)V$ if and only if for all $i\in\{2,\ldots,n\}$
\[
(1-\lambda)\cdot 1\leq 1-\lambda_i\leq (1+\lambda)\cdot 1.
\]
The above holds if and only if $|\lambda_i|\leq \lambda$ for all $i\in\{2,\ldots,n\}$, which is equivalent to $\lambda(H)\leq \lambda$. 
\end{proof}

\section{Proof of Lemma \ref{lem:symmequiv}}
\label{app:symmequiv}
\begin{namedtheorem}[Lemma \ref{lem:symmequiv} (restated)]
Let $\tilde{L}$ and $L$ be symmetric PSD matrices. Then $\tilde{L}$ is a directed $\epsilon$-approximation of $L$ if and only if $\tilde{L}\approx_{\epsilon}L$.
\end{namedtheorem}
\begin{proof}
The ``only if'' direction follows from Lemma \ref{lem:dir_implies_undir} by observing that $U\coloneqq(L+L^T)/2=L$ and $\tilde{U}\coloneqq (\tilde{L}+\tilde{L}^T)/2=\tilde{L}$.

For the ``if'' direction, suppose $\tilde{L}\approx_{\epsilon}L$. Equivalently,
\[
-\epsilon\cdot L \preceq \tilde{L}-L\preceq \epsilon\cdot L.
\]
Multiplying on the left and right by $L^{\dagger/2}$ preserves the ordering and yields
\[
-\epsilon\cdot L^{\dagger/2}LL^{\dagger/2} \preceq L^{\dagger/2}(\tilde{L}-L)L^{\dagger/2}\preceq \epsilon\cdot L^{\dagger/2}LL^{\dagger/2}.
\]
For each nonzero eigenvalue $\lambda$ of $L$, $L^{\dagger/2}$ has corresponding eigenvalue $1/\sqrt{\lambda}$ and hence $\left\|L^{\dagger/2}LL^{\dagger/2}\right\|\leq 1$. Since $L$ is symmetric we have $U\coloneqq(L+L^T)/2 = L$. Combining this with the above gives 
\begin{align*}
 \left\|U^{\dagger/2}(\tilde{L}-L)U^{\dagger/2}\right\|&= \left\|L^{\dagger/2}(\tilde{L}-L)L^{\dagger/2}\right\|\\
&\leq \epsilon\cdot\left\|L^{\dagger/2}LL^{\dagger/2}\right\|\\
&\leq \epsilon.
\end{align*}
Since $\tilde{L}\approx_{\epsilon}L$ by assumption, we must have $\mathrm{ker}(L)=\mathrm{ker}(\tilde{L})$. It follows that $\mathrm{ker}(\tilde{L}-L)=\mathrm{ker}(L)$ and $\mathrm{ker}((\tilde{L}-L)^T)=\mathrm{ker}(L)$ since $L$ and $\tilde{L}$ are symmetric. Since $U=L$, we have 
\begin{align*}
\mathrm{ker}(U)&=\mathrm{ker}(L)\\
&=\mathrm{ker}(\tilde{L}-L)\\
&=\mathrm{ker}(\tilde{L}-L)\cap\mathrm{ker}((L-\tilde{L})^T)
\end{align*}
\end{proof}

\section{Proof of Lemma \ref{lem:sq}}
\label{app:sq_preserves_approx}
The proof of Lemma \ref{lem:sq} is adapted from Cheng, Cheng, Liu, Peng, and Teng \cite{cheng2015}, which uses ideas from the work of Miller and Peng \cite{miller2013approximate}. We present these arguments here for the sake of completeness. 
\begin{claim}
\label{claim1}
Let $N$ and $\tilde{N}$ be symmetric matrices. If $I-\tilde{N}\approx_{\epsilon}I-N$ and $N$ is PSD then $I+\tilde{N}\approx_{\epsilon}I+N$
\end{claim}
\begin{proof}
By definition, the hypothesis is equivalent to 
\[
-\epsilon\cdot (I-N)\preceq(I-\tilde{N})-(I-N)\preceq \epsilon\cdot(I-N),
\]
which is equivalent to 
\[
-\epsilon\cdot (I-N)\preceq(I+\tilde{N})-(I+N)\preceq \epsilon\cdot(I-N).
\]
Since $N$ is PSD, we have $I-N\preceq I+N$. Thus we have
\[
-\epsilon\cdot(I+N)\preceq (I+N)-(I+\tilde{N})\preceq\epsilon\cdot (I+N),
\]
which is equivalent to $I+\tilde{N}\approx_{\epsilon}I+N$.
\end{proof}
\begin{claim}
\label{claim2}
If $N$ and $\tilde{N}$ are symmetric matrices such that $I-\tilde{N}\approx_{\epsilon}I-N$ and $I+\tilde{N}\approx_{\epsilon}I+N$ then
\[
\left[
\begin{array}{c c}
I & -\tilde{N} \\
-\tilde{N} & I
\end{array}
\right] \approx_{\epsilon}\left[
\begin{array}{c c}
I & -N \\
-N & I
\end{array}
\right]
\]
\end{claim}
\begin{proof}
We follow the proof of Lemma 4.4 in \cite{cheng2015}. For all vectors $x,y$ and symmetric matrices $U$, 
\[
\left[
\begin{array}{c}
x \\
y
\end{array}\right]^{T}\left[
\begin{array}{c c}
I & -U \\
-U & I
\end{array}
\right]\left[
\begin{array}{c}
x \\
y
\end{array}\right] = \frac{1}{2}\cdot\left((x+y)^{T}(I-U)(x+y) + (x-y)^T(I+U)(x-y) \right).
\]
The claim follows by taking $U=N$ and $U=\tilde{N}$ and our assumptions that $I-\tilde{N}\approx_{\epsilon}I-N$ and $I+\tilde{N}\approx_{\epsilon}I+N$.
\end{proof}
\begin{claim}
\label{claim3}
Let $N$ be a symmetric matrix. Then
\[
x^T(I-N^{2})x = \min_{y}\left[
\begin{array}{c}
x \\
y
\end{array}\right]^{T}\left[
\begin{array}{c c}
I & -N \\
-N & I
\end{array}
\right]\left[
\begin{array}{c}
x \\
y
\end{array}\right] 
\]
\end{claim}
\begin{proof}
This claim follows from Lemma B.2 in Miller and Peng \cite{miller2013approximate}, where it is stated in terms of the Schur complement. We modify the argument here without appealing to the Schur complement. 
\begin{align*}
\left[
\begin{array}{c}
x \\
y
\end{array}\right]^{T}\left[
\begin{array}{c c}
I & -N \\
-N & I
\end{array}
\right]\left[
\begin{array}{c}
x \\
y
\end{array}\right]  &= \|x\|^2 + \|y\|^2 - 2\cdot  \langle {y,Nx}\rangle\\
&\geq \|x\|^2 + \|y\|^2 - 2\cdot \|y\|\cdot \|Nx\| \\
&\geq \|x\|^2 - \|Nx\|^2 \\
&=x^{T}(I-N^2)x.
\end{align*}
The first inequality follows from Cauchy-Schwarz and is tight if and only if $y=c\cdot Nx$ for some scalar $c$. The second inequality follows from the fact that $z^2-2\cdot\|Nx\|\cdot z$ is minimized at $z=\|Nx\|$. So equality is achieved when $y=Nx$.
\end{proof}
Now we can prove Lemma \ref{lem:sq}.
\begin{namedtheorem}[Lemma \ref{lem:sq} (restated)]
Let $N$ and $\tilde{N}$ be symmetric matrices such that $I-\tilde{N}\approx_{\epsilon} I-N$ and $N$ is PSD, then $I-\tilde{N}^2 \approx_{\epsilon} I-N^2$.
\end{namedtheorem}
\begin{proof}
Let 
$
\mathbf{N}\coloneqq \left[
\begin{array}{c c}
I & -N \\
-N & I
\end{array}
\right] 
$
and
$
\tilde{\mathbf{N}}\coloneqq \left[
\begin{array}{c c}
I & -\tilde{N} \\
-\tilde{N} & I
\end{array}
\right] .
$ From our assumptions, it follows from Claim \ref{claim1} that $I+\tilde{N}\approx_{\epsilon}I+N$ and hence that $\tilde{\mathbf{N}}\approx_{\epsilon}\mathbf{N}$ by Claim \ref{claim2}. Combining this with Claim \ref{claim3} gives that for all vectors $x$,
\begin{align*}
(1-\epsilon)\cdot x^T(I-N)x&\leq (1-\epsilon)\cdot \left[
\begin{array}{c}
x \\
\tilde{N}x
\end{array}\right]^{T}\mathbf{N}\left[
\begin{array}{c}
x \\
\tilde{N}x
\end{array}\right]\\
&\leq \left[
\begin{array}{c}
x \\
\tilde{N}x
\end{array}\right]^{T}\tilde{\mathbf{N}}\left[
\begin{array}{c}
x \\
\tilde{N}x
\end{array}\right]\\
&=x^{T}(I-\tilde{N}^2)x.
\end{align*}
Similarly:
\begin{align*}
(1+\epsilon)\cdot x^T(I-N^2)x&= (1+\epsilon)\cdot \left[
\begin{array}{c}
x \\
Nx
\end{array}\right]^{T}\mathbf{N}\left[
\begin{array}{c}
x \\
Nx
\end{array}\right]\\
&\geq \left[
\begin{array}{c}
x \\
Nx
\end{array}\right]^{T}\tilde{\mathbf{N}}\left[
\begin{array}{c}
x \\
Nx
\end{array}\right]\\
&\geq x^{T}(I-\tilde{N}^2)x.
\end{align*}
So $I-\tilde{N}^2\approx_{\epsilon}I-N^{2}$.
\end{proof}
\end{document}